\documentclass[11pt]{article}

\usepackage{fullpage}

\usepackage[usenames,dvipsnames]{xcolor}
\usepackage[colorlinks,citecolor=blue,linkcolor=BrickRed]{hyperref}
\usepackage{makeidx} 
\usepackage{graphicx,tipa}
\usepackage{arcs,lmodern,fix-cm}
\usepackage{times}
\usepackage{amsfonts,latexsym,graphicx,epsfig,amssymb,color}
\usepackage{mathdots,amsmath,amsthm,amstext,setspace}
\usepackage{amsmath,amstext,amsthm,url,setspace, enumerate}
\usepackage{slashbox,multirow}
\usepackage{rotating}
\usepackage{verbatim}

\usepackage{float}
\usepackage{graphicx}   
\usepackage{xcolor}
\usepackage{microtype}
\usepackage{soul}
\usepackage{algorithm2e}
\usepackage{tabularx}
\usepackage{placeins}
\usepackage{mathtools}

\SetKwBlock{Repeat}{repeat}{end}

\newtheorem{theorem}{Theorem}[section]

\newtheorem{proposition}[theorem]{Proposition}
\newtheorem{lemma}[theorem]{Lemma}
\newtheorem{corollary}[theorem]{Corollary}

\newtheorem{observation}[theorem]{Observation}

\theoremstyle{definition}

\renewenvironment{proof} {\noindent\textbf{Proof:} } { \qed \medskip}

\newcommand{\reals}{\mathbb{R}}

\newcommand{\ignore}[1]{}

\newcommand{\ec}[1]{$\textsc{EE}_{#1}$}
\newcommand{\ecb}[2]{$\textsc{EE}_{#1}^{#2}$}

\newcommand{\naive}[1]{$\mathcal{N}_{#1}$}

\newcommand{\eenergy}[1]{E\left({#1}\right)}
\newcommand{\ttime}[1]{T\left( {#1} \right) }

\newcommand{\nlp}[2]{\textsc{NLP}_{#1}^{#2} }

\newcommand{\algo}[1]{\mathcal{A}\left(#1\right)}

\begin{document}
\title{\bf 
Energy Consumption of Group Search on a Line
\footnote{This is the full version of the paper with the same title which will appear in the proceedings of the 
46th International Colloquium on Automata, Languages and Programming
8-12 July 2019, Patras, Greece}
}

\author{
Jurek Czyzowicz\footnotemark[1] 
\and
Konstantinos Georgiou\footnotemark[2] 
\and
Ryan Killick\footnotemark[4] 
\and 
Evangelos Kranakis\footnotemark[4] 
\and
Danny Krizanc\footnotemark[7]
\and
Manuel Lafond\footnotemark[5] 
\and
Lata Narayanan\footnotemark[8] 
\and
Jaroslav Opatrny\footnotemark[8]
\and
Sunil Shende\footnotemark[9]
}

\def\thefootnote{\fnsymbol{footnote}}
\footnotetext[1]{
Universite du Québec en Outaouais, Gatineau, Qu\'{e}bec, Canada, \texttt{jurek.czyzowicz@uqo.ca}
}
\footnotetext[2]{
Dept. of Mathematics, 
Ryerson University, 
Toronto, ON, Canada, \texttt{konstantinos@ryerson.ca}
}
\footnotetext[4]{
School of Computer Science, Carleton University, Ottawa ON, Canada, \texttt{ryankillick,kranakis@scs.carleton.ca}
}
\footnotetext[7]{
Department of Mathematics \& Comp. Sci., Wesleyan University, Middletown, CT, USA, \texttt{dkrizanc@wesleyan.edu}
}
\footnotetext[8]{
Department of Comp. Sci. and Software Eng., Concordia University, Montreal, Qu\'{e}bec,  Canada, \texttt{lata,opatrny@encs.concordia.ca}
}
\footnotetext[9]{
Department of Computer Science, Rutgers University, Camden, USA,
\texttt{shende@camden.rutgers.edu}
}
\footnotetext[5]{
Department of Computer Science, Universit\'{e} de Sherbrooke, Qu\'{e}bec,  Canada
\texttt{Manuel.Lafond@usherbrooke.ca}
}
\maketitle

\begin{abstract}
Consider two robots that start at the origin of the infinite line in search of an \textit{exit} at an unknown location on the line. The robots can collaborate in the search, but can only communicate if they arrive at the same location at exactly the same time, i.e. they use the so-called {\em face-to-face} communication model.  The group search time is defined as the worst-case time as a function of $d$, the distance of the exit from the origin, when both robots can reach the exit. It has long been known that for a single robot traveling at unit speed, the search time is at least $9d - o(d)$; a simple doubling strategy achieves this time bound. It was shown recently in \cite{chrobak2015group} that $k \geq 2$ robots traveling at unit speed also require at least $9d$ group search time. 

We investigate \textit{energy-time trade-offs} in group search by two robots, where the energy loss experienced by a robot traveling a distance $x$ at constant speed $s$ is given by $s^2 x$ and is motivated by first principles in physics and engineering. Specifically, we consider the problem of minimizing the total energy used by the robots, under the constraints that the search time is at most a multiple $c$ of the distance $d$ and the speed of the robots is bounded by $b$. Motivation for this study is that for the case when robots must complete the search in $9d$ time with maximum speed one ($b=1;~c=9$), a single robot requires at least $9d$ energy, while for two robots, all previously proposed algorithms consume at least $28d/3$ energy.

When the robots have bounded memory and can use only a constant number of fixed speeds, we generalize an algorithm described in \cite{BS95,chrobak2015group} to obtain a family of algorithms parametrized by pairs of $b,c$ values that can solve the problem for the entire spectrum of these pairs for which the problem is solvable. In particular, for each such pair, we determine optimal (and in some cases nearly optimal) algorithms inducing the lowest possible energy consumption. 

We also propose a novel search algorithm that simultaneously achieves search time $9d$ and consumes energy $8.42588d$. Our result shows that two robots can search on the line in optimal time $9d$ while consuming less total energy than a single robot within the same search time. 
Our algorithm uses robots that have unbounded memory, and a finite number of dynamically computed speeds. It can be generalized for any $c, b$ with $cb=9$, and consumes energy $8.42588b^2d$. 

\end{abstract}

\section{Introduction}
\label{sec: intro}


The problem of searching for a treasure at an unknown location in a specified continuous domain was initiated over fifty years ago ~\cite{beck1964linear,bellman1963optimal}.  Search domains that have been considered include the infinite line \cite{baezayates1993searching,beck1964linear,bellman1963optimal,kao1996searching}, a set of rays \cite{Bose13,BrandtFRW17}, the unit circle \cite{Watten2017,CGKNOV15,pattanayak2017evacuating}, and  polygons \cite{CKKNOS,HIKK01,K94}. Consider a robot (sometimes called a mobile agent) starting at some known location in the domain and looking for \textit{an exit} that is located at an unknown distance $d$ away from the start. What algorithm should the robot use to find the exit as soon as possible? The most common cost measure used for the search algorithm is the worst-case {\em search time},  as a function of the distance $d$ of the exit from the starting position. For a fixed-speed robot, the search time  is proportional to the length of the trajectory of the robot.  Other measures such as turn cost~\cite{demaine2006online} and different costs for revisiting \cite{Bose16} have been sometimes considered. 

We consider for the first time the {\em energy} consumed by the robots while executing the search algorithm. The energy used by a robot to travel a distance $x$ at speed $s$ is computed as $s^2x$ and is motivated from the concept of viscous drag in fluid dynamics; see Section~\ref{sec: definition} for details on the energy model. For a single robot searching on the line, the classic \textit{Spiral Search} algorithm (also known as the doubling strategy) has search time $9d$ and is known to be \textit{optimal} when the robot moves with unit speed.  Since in the worst case, the robot travels distance $9d$ at unit speed, the energy consumption is $9d$ as well. Clearly, as the speed of the robot increases, the time to find the exit decreases but at the same time, the energy used increases. Likewise, as the speed of the robot decreases, the time to find the exit increases, while the energy consumption decreases. Thus there is a natural trade-off between the time taken by the robot to search for the exit and the energy consumed by the robot. To investigate this trade-off, we consider the problem of minimizing the total energy used by the robots to perform the search when the speed of the robot is bounded by $b$, and the time for the search is at most a multiple $c$ of the distance $d$ from the starting point to the exit.

Group search by a set of $k \geq 2$ collaborating robots has recently gained a lot of attention. In this case, the search time is the time when {\em all} $k$ robots reach the exit. The problem has also been called {\em evacuation}, in view of the application when it is desired that all robots reach and evacuate from the exit. 
Two models of communication between the robots have been considered. In the wireless communication model, the 
robots can instantly communicate with each other at any time and over any distance. In the face-to-face communication model (F2F), two robots can communicate only when in the same place at the same time.  In many search domains, and for both communication models, group search by $k \geq 2$ agents has been shown to take less time than search by a single agent; see for example \cite{CGKNOV15,CKKNOS}. 

In this paper, we focus on group search on the line, by two robots using the F2F model. 
 Chrobak {\em et al}  \cite{chrobak2015group} showed that group search in this setting cannot be performed in time less than $9d - o(d)$, regardless of the number of robots, assuming all robots use at most unit speed. They also describe several strategies that achieve search time $9d$. In the first strategy, the two robots independently perform the Spiral Search algorithm, using unit speed during the entire search. Next, they consider a strategy first described in \cite{BS95},  that we call the {\em Two-Turn} strategy, whereby two robots head off independently in opposite directions at speed $1/3$; when one of them finds the exit, it moves at unit speed to chase and catch the other robot, after which they both return at unit speed to the exit. Finally, they present a new strategy, called the {\em Fast-Slow} algorithm in which one robot moves at unit speed, while the other  robot moves at speed 1/3, both performing a spiral search. The doubling strategy is very energy-inefficient, it uses energy $18d$ if the two robots always travel together, or  $14D$  if the robots start by moving in opposite directions. The other two algorithms both use energy $28d/3 > 9d$. Interestingly, the two strategies that achieve an energy consumption of $28d/3$ with search time $9d$, both use two different and pre-computed speeds, but are quite different in terms of the robot capacities needed. In the Two-Turn strategy, the robots are extremely simple and use constant memory; they use only three states. In  Fast-Slow and  Spiral Search,  the robots need unbounded memory, and perform computations to determine how far to go before turning and moving in the opposite direction. 

Memory capability, time- and speed-bounded search, and energy consumption by a two-robot group search algorithm on the line: these considerations motivate the following questions that we address in our paper:
\begin{enumerate}
\item Is there a search strategy for constant-memory robots that has energy consumption $<9d$?
\item Is there \textit{any} search strategy that uses time $9d$ and energy $< 9d$?  
\end{enumerate}

\subsection{Our results}

We generalize the Two-Turn strategy for any values of $c, b$. We analyze the entire spectrum of values of $c,b$ for which the problem admits a solution, and for each of them we provide optimal (and in some cases nearly optimal) speed choices for our  robots (Theorem~\ref{thm: (sub)optimal speeds naive bounded speeds}). In particular, and somewhat surprisingly, our proof makes explicit how for any fixed $c$ the optimal speed choices do not simply "scale" with $b$; rather more delicate speed choices are necessary to comply with the speed and search time bounds. For the special case of $c\cdot b=9$, our results match with the specific Two-Turn strategy described in \cite{chrobak2015group}. Our results further show that no Two-Turn strategy can achieve energy consumption less than $9d$ while keeping the search time at $9d$. In fact, we conjecture that this trade-off is impossible for \textit{any} group search strategy that uses only constant memory robots.

In the unbounded-memory model, for the special case of $c=9$ and $b=1$, we give a novel search algorithm that achieves energy consumption of $8.42588d$, thus answering the second question above in the affirmative. This result shows that though two robots cannot search faster than one robot on the line \cite{chrobak2015group},  somewhat surprisingly, two robots can search using less total energy than one robot, in the same optimal time.  Our algorithm uses robots that have unbounded memory, and a finite number of dynamically computed speeds. Note that our algorithm can be generalized for any $c, b$ with $cb=9$, and utilizes energy $8.42588b^2d$ (Theorem~\ref{thm: CR non-primitive bounded speeds}).

\subsection{Related Work} 
\label{sec: related work}

Several authors have investigated various aspects of mobile robot (agent) search, resulting in an extensive literature on the subject in theoretical computer science and mathematics (e.g., see \cite{alpern2003theory,GAL} for reviews). Search by constant-memory robots has been done mainly for finite-state automata (FSA) operating in discrete environments like infinite grids, their finite-size subsets (labyrinths) and other graphs. The main concern of this research was the feasibility of search, rather than time or energy efficiency. For example, \cite{Budach} showed that no FSA can explore all labyrinths, while \cite{Blum} proved that one FSA using two pebbles or two FSAs, communicating according to the F2F model can explore all labyrinths.  However, no collection of FSAs may explore all finite graphs communicating in the F2F model~\cite{Rol80} or wireless model~\cite{Cook}. On the other hand, all graphs of size $n$ may be explored using a robot having $O(\log n)$ memory~\cite{Reingold}.

Exploration of infinite grids is known as the ANTS problem~\cite{ELSUW}, where it was shown that four collaborating FSAs in the semi-synchronous execution model and communicating according to the F2F scenario can explore an infinite grid. Recently, \cite{BrandtUW18} showed that four FSAs are really needed to explore the grid (while three FSAs can explore an infinite band of the  2-dimensional grid).

Continuous environment cases 
have been investigated in several papers when the efficiency of the search  is often represented by the time of reaching the target (e.g., see \cite{baezayates1993searching,beck1964linear,bellman1963optimal,kao1996searching}). Even in the case of continuous environment as simple as the infinite line, after the seminal papers~\cite{beck1964linear,bellman1963optimal}, various scenarios have been studied where the turn cost has been considered~\cite{demaine2006online}, the environment was composed of portions permitting different search speeds~\cite{CzyzowiczKKNOS17}, some knowledge about the target distance was available~\cite{Bose13} or where some other parameters are involved in the computation of the cost function~\cite{Bose16} (e.g. when the target is moving). 

The group search, sometimes interpreted as the evacuation problem has been studied first for the disc environment under the F2F~\cite{Watten2017,CGGKMP,CGKNOV15,CKKNOS,LMS} and wireless \cite{CGGKMP} communication scenarios and then also for other geometric environments (e.g., see \cite{CKKNOS}). Other variants of search/evacuation problems with a combinatorial flavour have been recently considered in \cite{CGS18,CGKKKNOS18a,CGKKKNOS18b,GeorgiouKK16,GeorgiouKK17}. Some papers investigated the line search problem in the presence of crash faulty~\cite{PODC16} and Byzantine faulty agents~\cite{isaacCzyzowiczGKKNOS16}. The interested reader may also consult the recent survey~\cite{CGKMAC19} on selected search and evacuation topics.

The energy used by a mobile robot is usually considered as being spent solely for travelling. As a consequence, in the case of a single, constant speed robot the search time is proportional to the distance travelled and the energy used by a robot. Therefore the problems of minimization of time, distance or energy are usually equivalent for most robots' tasks. For teams of collaborating robots, the searchers often need to synchronize their walks in order to wait for information communicated by other searchers (e.g, see~\cite{Watten2017,CGGKMP,LMS}), hence the time of the task and the distance travelled are different. However, the distance travelled by a robot and its energy used are still commensurable quantities.

To the best of our knowledge, energy consumption as a function of mobile robot speed which is based on natural laws of physics (related to the {\em drag force}) has never been studied in the search literature before. Our present work is motivated by~\cite{chrobak2015group}, which proves that the competitive ratio $9$ is tight for group search time with two mobile agents in the F2F model when both agents have unit maximal speeds. More exactly, it follows from~\cite{chrobak2015group} that having more unit-speed robots cannot improve the group search time obtained by a single robot. Nevertheless, our paper shows that using more robots can improve the energy spending, while keeping the group-search time still the best possible.

~\cite{chrobak2015group} presents interesting examples of group search algorithms for two distinct speed robots communicating according to the F2F scenario. An interested reader may consult ~\cite{SIROCCO16}, where optimal group search algorithms for a pair of distinct maximal speed robots were proposed for both communication scenarios (F2F and wireless) and for any pair of robots' maximal speeds. It is interesting to note that, according to~\cite{SIROCCO16}, for any distinct-speed robots with F2F communication, the optimal group search time is obtained only if one of the robots perform the search step not using its full speed. 

\textbf{Paper Organization:}
In Section~\ref{sec: definition} we formally define the evacuation problem \ecb{c}{b}, and  proper notions of efficiency. 
Our algorithms and their analysis for constant-memory robots is presented in Section~\ref{sec: naive algo}, while in Section~\ref{sec: better algo} we introduce and analyze algorithms for unbounded-memory robots. 
All omitted proofs can be found in the Appendix. 
Also, due to space limitations, all figures appear in Appendix~\ref{sec: omitted figures}.

\section{Preliminaries}
\label{sec: definition}

Two robots start walking from the origin of an infinite (bidirectional) line in search of a hidden exit at an unknown absolute distance $d$ from the origin. The exit is considered found only when one of the robots walks over it. An algorithm for group search by two robots specifies trajectories for both robots and terminates when both robots reach the exit.  The time by which the second robot reaches the exit is referred to as the \textit{search  time} or the \textit{evacuation time}. 

\noindent {\bf Robot models:}  The two robots operate under the F2F communication model in which two robots can communicate only when they are in the same place at the same time. Each robot  can change its speed at any time. We distinguish between {\em constant-memory} robots that can only travel at a constant number of hard-wired speeds, and {\em unbounded-memory} robots that can dynamically compute speeds and distances, and travel at any possible speed. 

\noindent {\bf Energy model:} A robot moving at constant speed $s$ traversing an interval of length $x$ is defined to use energy $s^2 \cdot x$. This model is well motivated from first principles in physics and engineering and corresponds to the energy loss experienced by an object moving through a viscous fluid \cite{batchelor2000}. In particular, an object moving with constant speed $s$ will experience a {\em drag} force $F_{D}$ proportional\footnote{The constant of proportionality has (SI) units $kg/m$ and depends, among other things, on the shape of the object and the density of the fluid through which it moves.} to $s^2$. In order to maintain the speed $s$ over a distance $x$ the object must do work equal to the product of $F_D$ and $x$ resulting in a continuous energy loss proportional to the product of the object's squared speed and travel distance. For simplicity we have taken the proportionality constant to be one.  

The total energy that a robot uses traveling at speeds $s_1, s_2, \ldots, s_t$, traversing intervals $x_1, x_2, \ldots, x_t$, respectively, is defined as $\sum_{i=1}^t s_i^2 \cdot x_i$. 
For  group search with two  robots, the \textit{energy consumption} is defined as the sum total of the two robots' energies used until the search algorithm terminates.  

For each $d>0$ there are two possible locations for the exit to be at distance $d$ from the origin: we will refer to either of these as input instances $d$ for the group search problem. Our goal is to solve the following \textit{optimized} search problem parametrized by two values, $\mathbf{b}$ and $\mathbf{c}$:\\
 
\noindent
{\textcolor{darkgray}{$\blacktriangleright$}\nobreakspace\sffamily\bfseries Problem \ecb{c}{b}}: Design a group search algorithm for two robots in the F2F model that minimizes the energy consumption for $d$-instances under the constraints that the search time is no more than $\mathbf{c\cdot d}$ and the robots use 
speeds that are at most $\mathbf{b}$. When there are no speed limits on the robots  (i.e.  $b=\infty$), we abbreviate \ecb{c}{\infty} by \ec{c}. Note that $b,c$ are inputs to the algorithm, but $d$ and the exact location of the exit are not known.

As it is standard in the literature on related problems, we assume that the exist is at least a known constant distance away from the origin. In this work, we pick the constant equal to 2, although our arguments can be adjusted to any other constant. 
It is not difficult to show that \ecb{c}{b} is well defined for each $b,c>0$ with $bc\geq1$, and the optimal offline solution, for instance $d$, is for both robots to move at speed $\frac{1}{c}$ to the exit. This offline algorithm has energy consumption $\frac{2d}{c^2}$ (see Observation~\ref{obs: opt offline bounded speeds} in Appendix~\ref{sec: opt offline bounded speeds}). 
Consider an online algorithm for \ecb{c}{b}, which on any instance $d$ has energy consumption at most $e(c,b,d)$. The \textit{competitive ratio} of the algorithm is defined as 
$
\sup_{d>0} \tfrac{c^2}{2d}~ e(c,b,d).
$

Due to~\cite{chrobak2015group}, and when $b=1$, no online algorithm (for two robots) can have evacuation time less than $9d-\epsilon$ (for any $\epsilon>0$ and for large enough $d$). By scaling, using arbitrary speed limit $b$, we obtain the following fact. 

\begin{observation}
\label{lem: restrictions on cb f2f}
No online F2F algorithm can solve \ecb{c}{b} if $cb<9$.
\end{observation}

\section{Solving \ecb{c}{b} with Constant-Memory Robots}
\label{sec: naive algo}

In this section we propose a family of algorithms for solving \ecb{c}{b} (including $b=\infty$). 
The family uses an algorithm that is parametrized by three discrete speeds: $s$, $r$ and $k$. The robots use these speeds depending on  finite state control as follows:\\

\noindent
{\textcolor{darkgray}{$\blacktriangleright$}\nobreakspace\sffamily\bfseries Algorithm \naive{s,r,k}: }
Robots start moving in opposite directions with speed $s$ until the exit is found by one of them. The finder changes direction and moves at speed $r > s$ until it catches the other robot. Together the two 
robots return to the exit using speed $k$. 

\begin{lemma}[Proof on page~\pageref{sec: energy time formulas}]
\label{lem: energy time formulas}
Let $b,c$ be such that there exist $s,r,k$ for which \naive{s,r,k} is feasible. 
Then, for instance $d$ of \ecb{c}{b}, the induced evacuation time of \naive{s,r,k} is 
$
d\cdot \ttime{s,r,k}
$
and the induced energy consumption is 
$
2d\cdot \eenergy{s,r,k}
$, where
$$
\ttime{s,r,k}:= \frac{2 (k+r)}{k (r-s)}+\frac{1}{s}, ~~
\eenergy{s,r,k}:=\frac{r}{r-s}
\left(
s^2+r^2+2k^2
\right)
$$
\end{lemma}

\ignore{
\begin{proof}[Proof of Lemma~\ref{lem: energy time formulas}]
Consider the moment that the exit is located, after time $d/s$ time of searching. 
The robot that now chases the other speed-$s$ robot at constant speed $r>s$ will reach it after $2d/(r-s)$ time. To see this note that the configuration is equivalent to that the speed-$s$ robot is immobile and the other robot moves at speed $r-s$, having to traverse a total distance of $2d$. Moreover, the speed-$s$ robot traverses an additional length $2ds/(r-s)$ segment till it is caught, being a total of $2ds/(r-s)+2d$ away from the exit. Once robots meet, the walk to the exit at speed $k$, which takes additional time $(2ds/(r-s)+2d)/k$. Overall the evacuation time equals
$$
\frac{d}{s}+\frac{2d}{r-s}+\frac{2ds/(r-s)+2d}{k} =
d\left( \frac{2 (k+r)}{k (r-s)}+\frac{1}{s}\right).
$$

Similarly we compute the total energy till both robots reach the exit. The energy spent by the finder is
$$
d\cdot s^2 + \left( \frac{2ds}{r-s} + 2d \right) \cdot r^2 + \left( \frac{2ds}{r-s} + 2d \right) \cdot k^2,
$$
while the energy spent by the non finder is 
$$
\left(d + \frac{2ds}{r-s}\right) \cdot s^2  + \left( \frac{2ds}{r-s} + 2d \right) \cdot k^2.
$$
Adding the two quantities and simplifying gives the promised formula. 
\end{proof}
}


We propose a systematic way in order to find optimal values for $s,r,k$ of algorithm \naive{s,r,k} for optimization problem \ecb{c}{b} (including $b=\infty$), whenever such values exist. 

\begin{theorem}[Proof on page~\pageref{sec: naive as NLP}]
\label{thm: naive as NLP}
Algorithm \naive{s,r,k} gives rise to a feasible solution to problem \ecb{c}{b} if and only if $bc\geq 9$. For every such $b,c>0$, the optimal choices  of \naive{s,r,k} can be obtained by solving Non Linear Program:
\begin{align}
& ~~\min_{s,r,k \in \reals} \eenergy{s,r,k} \tag{$\nlp{c}{b}$}\\
s.t. ~~&
\ttime{s,r,k} \leq c  \notag \\
&r\geq s \notag \\
&0 \leq s,r,k\leq b \notag
\end{align}
where functions $\eenergy{\cdot,\cdot,\cdot}, \ttime{\cdot,\cdot,\cdot}$ are as in Lemma~\ref{lem: energy time formulas}.
Moreover, if $s_0, r_0, k_0$ are the optimizers to $\nlp{c}{b}$, then the competitive ratio of \naive{s_0,r_0,k_0} equals 
$
c^2 \cdot \eenergy{s_0,r_0,k_0}.
$
\end{theorem}

The following subsections are devoted to solving $\nlp{c}{b}$, effectively proving Theorem~\ref{thm: (sub)optimal speeds naive bounded speeds}.
First in Section~\ref{sec: naive for unbounded} we solve the case $b=\infty$ and we use our findings to solve the case of bounded speeds $b$ in the follow-up Section~\ref{sec: naive for bounded}.

\subsection{Optimal Choices of \naive{s,r,k} for the Unbounded-Speed Problem}
\label{sec: naive for unbounded}

In this section we propose solutions to the unbounded-speed problem \ec{c}. 
Since \ec{c} is the same as \ecb{c}{\infty}, by Observation~\ref{obs: opt offline bounded speeds}, the problem is well-defined for every fixed $c>0$. Moreover, 
by the proof of Theorem~\ref{thm: naive as NLP} (see proof of Lemma~\ref{lem: optimizer exists} in the Appendix) algorithm \naive{s,r,k} induces a feasible solution  for every $c>0$ as well, and the optimal speeds can be found by solving $\nlp{c}{\infty}$.
Indeed, in the remaining of the section we show how to choose optimal values for $s,r,k$ for solving \ec{c} with \naive{s,r,k}.
Let 
\begin{equation}
\label{equa: values sigma, rhow, kappa}
\sigma \approx  2.65976, \rho \approx 11.3414 , \kappa \approx 6.63709,
\end{equation}
whose exact values are the roots of an algebraic system and will be formally defined later. 
The main theorem of this section reads as follows. 

\begin{theorem}[Proof on page~\pageref{sec: optimal speeds naive unbounded speeds}]
\label{thm: optimal speeds naive unbounded speeds}
Let $\sigma, \rho, \kappa$ as in~\eqref{equa: values sigma, rhow, kappa}.  
For every $c>0$, the optimal speeds of \naive{s,r,k} for problem \ec{c} are
$
s=\tfrac{\sigma}{c}, ~~r=\tfrac{\rho}{c}, ~~k=\tfrac{\kappa}{c}.
$
Moreover, the competitive ratio of the corresponding solution is independent of $c$ and equals $\frac{\rho  \left(2 \kappa ^2+\rho ^2+\sigma ^2\right)}{\rho -\sigma } \approx 292.369$. 
\end{theorem}

A high level outline of the proof of Theorem~\ref{thm: optimal speeds naive unbounded speeds} is as follows. First we show that any optimal choices of the speeds of \naive{s,r,k}\ must satisfy the time constraint of $\nlp{c}{\infty}$ tightly. Then, we show that finding optimal speeds $s,r,k$ of \naive{s,r,k}\ for the general problem \ec{c} reduces to problem \ec{1}. Finally, we obtain the optimal solution to $\nlp{1}{\infty}$ by standard tools of nonlinear programming (KKT conditions).

\subsection{(Sub)Optimal Choices of \naive{s,r,k} for the Bounded-Speed Problem}
\label{sec: naive for bounded}
In this section, we show how to choose optimal values for $s,r,k$ for solving \ecb{c}{b} with \naive{s,r,k}, for the entire spectrum of $c,b$ values for which the problem is solvable by online algorithms. 

The main result of this section is the following:

%


\begin{theorem}
\label{thm: (sub)optimal speeds naive bounded speeds}
Let $\gamma_1\approx 9.06609$, $\gamma_2=\rho\approx11.3414$, and $\sigma, \rho, \kappa$ as in~\eqref{equa: values sigma, rhow, kappa}.  
For every $c,b>0$ with $cb\geq 9$, the following choices of speeds $s,r,k$ are feasible for \naive{s,r,k}
$$
\begin{tabular}{l|ccc}
			& $9 \leq cb \leq \gamma_1$		&	$\gamma_1 < cb < \gamma_2$	&	$cb \geq \gamma_2$ \\
\hline
$s$ 	& 	$\frac{-\sqrt{(bc)^2-10 bc+9}+bc-3}{2 c}$	& 	$0.532412 b-0.0262661 b^2 c$	&	$\sigma/c$ \\
$r$	& $b$	& $b$		&	$\rho/c$ \\
$k$	& $b$	& $\frac{2 b s}{b c s-b-c s^2-s}$		&	$\kappa/c$ \\
\end{tabular}
$$
\ignore{
$$
s=
\left\{
\begin{array}{ll}
\frac{-\sqrt{(bc)^2-10 bc+9}+bc-3}{2 c}&, 9 \leq cb \leq \gamma_1 \\
aaaa &, ~\gamma_1 < cb < \gamma_2 \\
\sigma/c&, cb \geq \gamma_2
\end{array}
\right. , ~~
r=
\left\{
\begin{array}{ll}
b &, 9 \leq cb \leq \gamma_1 \\
aaaa &, ~\gamma_1 < cb < \gamma_2 \\
\rho/c&, cb \geq \gamma_2
\end{array}
\right. , ~~
k=
\left\{
\begin{array}{ll}
b&, 9 \leq cb \leq \gamma_1 \\
aaaa &, ~\gamma_1 < cb < \gamma_2 \\
\kappa/c&, cb \geq \gamma_2
\end{array}
\right. 
$$
}
The induced competitive ratio is given by: 
$$
f(x):=
\left\{
\begin{array}{ll}
\frac{1}{2} x \left(x \left(x-\sqrt{(x-9) (x-1)}\right)+\sqrt{(x-9) (x-1)}+3\right), & 9 \leq x \leq \gamma_1 \\
\frac{x^2 \left((0.532412\, -0.0262661 x)^2+\frac{11595.8 (20.2699\, -1. x)^2}{(x (x (x-2.46798)-398.916)+2221.18)^2}+1\right)}{0.0262661 x+0.467588}, 
& \gamma_1 < x < \gamma_2 \\
292.369 & x \geq\gamma_2
\end{array}
\right.
$$
and the induced energy, for instances $d$, is $f(cb)\frac{2d}{c^2}$. 
Moreover, the competitive ratio depends only on the product $cb$.

In particular, the speeds' choices are optimal when $cb\leq \gamma_1$ and when $cb\geq \gamma_2$. When $\gamma_1<cb<\gamma_2$, the derived competitive ratio is no more than 0.03 additively off from that induced by optimal choices of $s,r,k$. 
\end{theorem}

\begin{corollary}
For $c=9, b=1$, the bounded-memory robot algorithm \naive{s,r,k} has energy consumption $28d/3$ and competitive ratio 378.

\end{corollary}

Theorem~\ref{thm: (sub)optimal speeds naive bounded speeds} is proven by solving 
$\nlp{c}{b}$ of Theorem~\ref{thm: naive as NLP}. In particular, the induced competitive ratio of \naive{s,r,k} for the choices of Theorem~\ref{thm: (sub)optimal speeds naive bounded speeds} is summarized in Figure~\ref{fig: compratiof2f} (Appendix~\ref{sec: omitted figures}). Speed values $s,r,k$, are chosen optimally when $cb$ is either at most $\gamma_1$ or at least $\gamma_2$ (i.e. optimizers to $\nlp{c}{b}$ admit analytic description). The optimal speed parameters when $\gamma_1<cb<\gamma_2$ cannot be determined analytically (they are roots of high degree polynomials). The values that appear in Theorem~\ref{thm: (sub)optimal speeds naive bounded speeds} are heuristically chosen, but interestingly induce nearly optimal competitive ratio, see Figure~\ref{fig: comparisonVSoptimal-additive} (Appendix~\ref{sec: omitted figures}). 

\ignore{
\begin{minipage}[t]{0.39\textwidth}
  \vspace{0pt}  
\begin{figure}[H]
  \centering
  \includegraphics[width=0.8\linewidth]{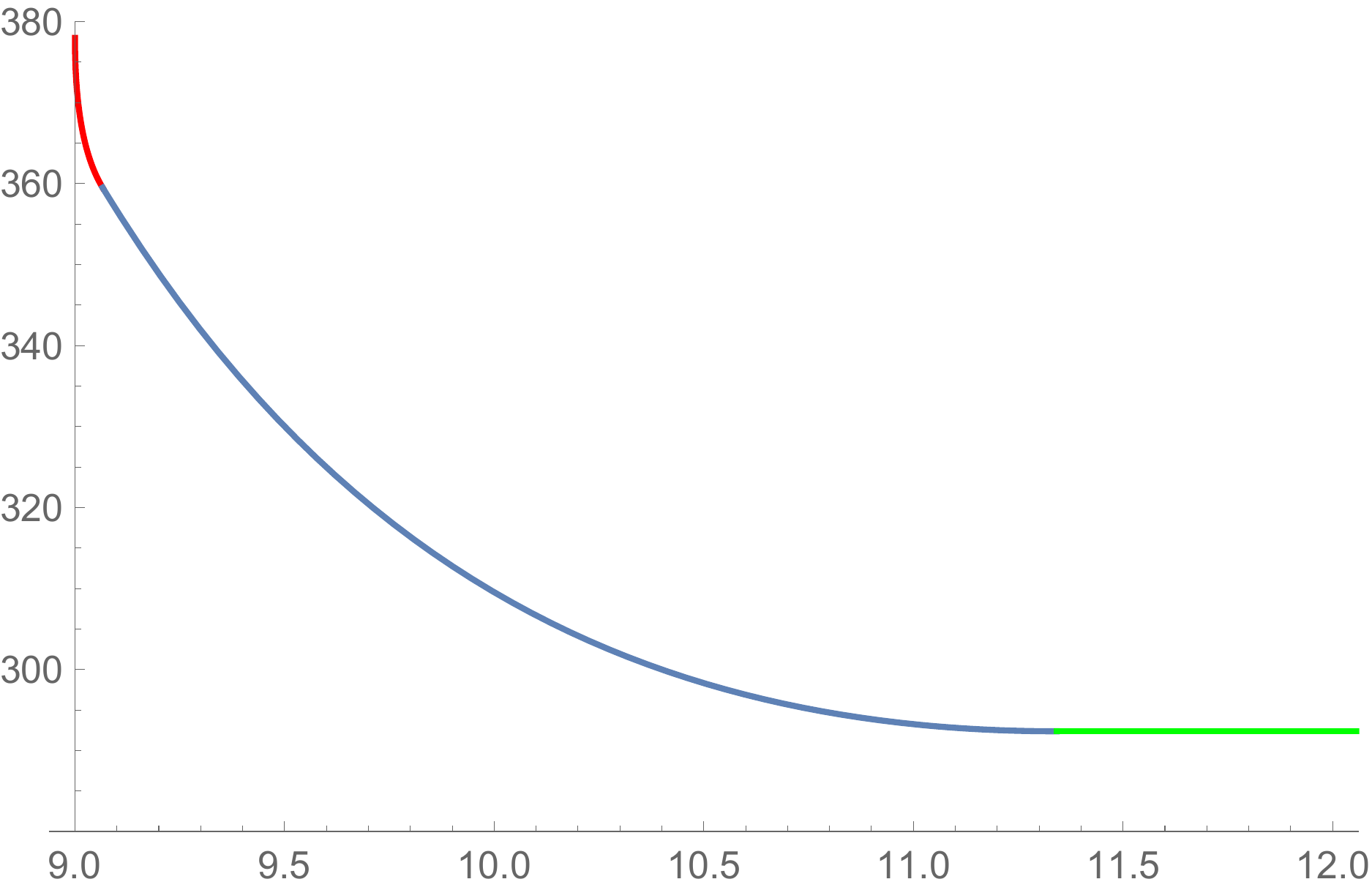}
\caption{
The competitive ratio of algorithm 
\naive{s,r,k} (vertical axis) for the entire spectrum of $cb\geq 9$ (horizontal axis). Red curve corresponds to the case $cb\leq \gamma_1$, blue curve to the case $cb \in (\gamma_1,\gamma_2)$ and green curve to the case $cb\geq \gamma_2$. The curve is continuous and differentiable for all $cb\geq 9$. 
}
\label{fig: compratiof2f}
\end{figure}
\vline
\end{minipage}
~~~~~
\begin{minipage}[t]{0.54\textwidth}
  \vspace{0pt}  
\begin{figure}[H]
  \centering
  \includegraphics[width=0.65\linewidth]{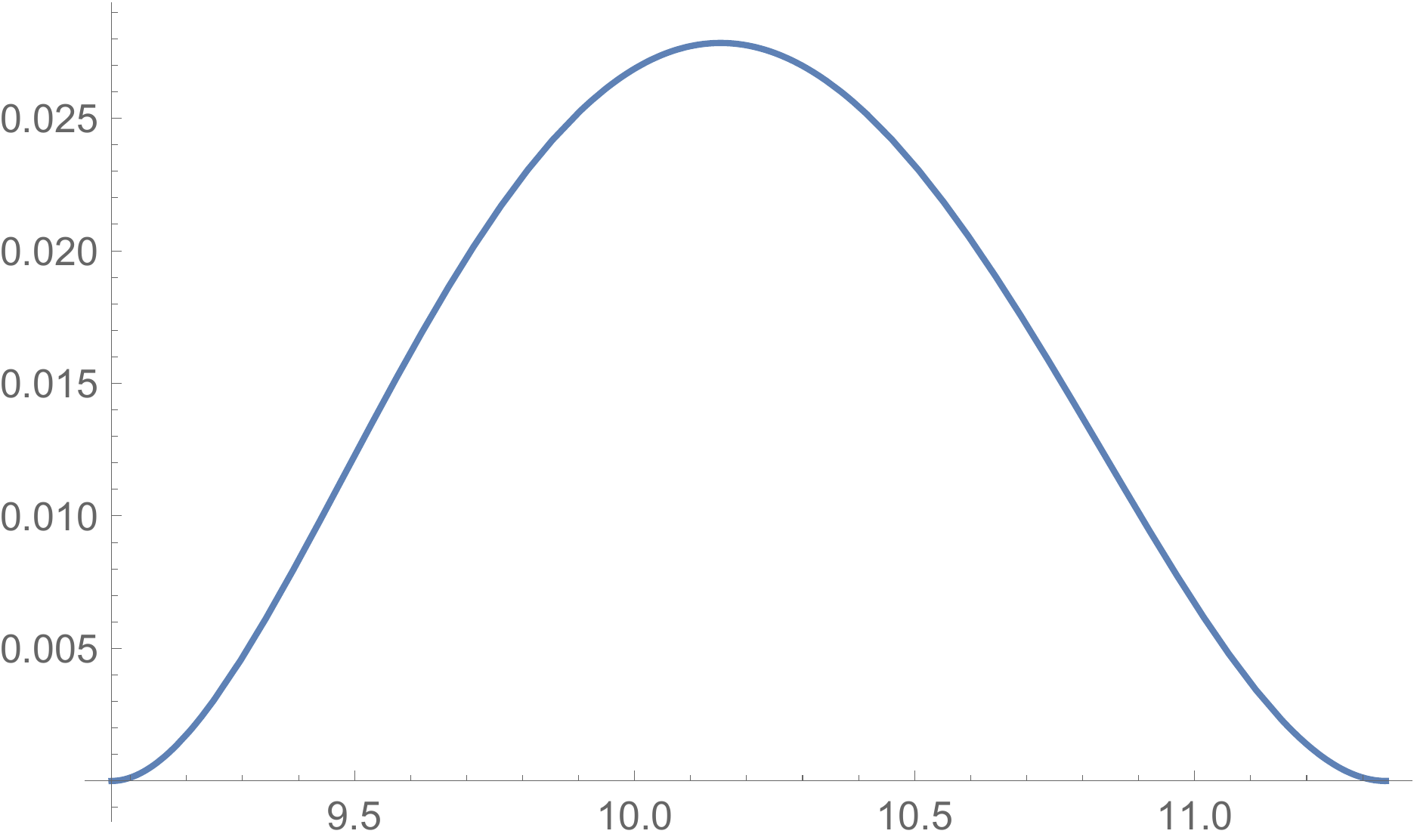}
\caption{
Comparison between the competitive ratio achieved by using \textit{optimal} speed parameters to $\nlp{c}{b}$ of Theorem~\ref{thm: naive as NLP} (calculated numerically using software) and the competitive ratio achieved by the choices of Theorem~\ref{thm: (sub)optimal speeds naive bounded speeds}. 
The vertical axis is the difference of the competitive ratios, and the horizontal axis corresponds to the values of $cb\in (\gamma_1,\gamma_2)$ (for all other values of $cb$ the difference is provably 0). 
}
\label{fig: comparisonVSoptimal-additive}
\end{figure}
\end{minipage}
}

The proof of Theorem~\ref{thm: (sub)optimal speeds naive bounded speeds} is given by Lemma~\ref{lem: optimizers for c>9/b, speed bound b} (the case $cb\leq \gamma_1$), Lemma~\ref{lem: optimal speeds naive bounded speeds b>=c>=r1} (the case $cb\geq \gamma_2$), and Lemma~\ref{lem: naive case between gamma1 gamma2} (the case $\gamma_1 < cb < \gamma_2$). Next we state these Lemmata, and we sketch their proofs.

\begin{lemma}[Proof on page~\pageref{sec: optimizers for c>9/b, speed bound b}]
\label{lem: optimizers for c>9/b, speed bound b}
For every $c\in (9/b, \gamma_1/b]$, where $\gamma_1\approx 9.06609$, the optimizers to~$\nlp{c}{b}$ are
$k=r=b$, and $s_b=\frac{-\sqrt{(bc)^2-10 bc+9}+bc-3}{2 c}$. 
The induced 
competitive ratio is 
$f(cb)$, (see definition of $f(x)$ for $x\leq \gamma_1$ in statement of Theorem~\ref{thm: (sub)optimal speeds naive bounded speeds}), and the energy consumption, for instances $d$, is $f(cb)\frac{2d}{c^2}$.
\ignore{
The induced energy is
then $\frac{b \left(b c \left(b c-\sqrt{(b c-9) (b c-1)}\right)+\sqrt{(b c-9) (b c-1)}+3\right)}{c}$, and the competitive ratio is 
$
\frac{1}{2} b c \left(b c \left(b c-\sqrt{(b c-9) (b c-1)}\right)+\sqrt{(b c-9) (b c-1)}+3\right)
$.
}
\ignore{
Later I call $t(x) = \sqrt{(x-9) (x-1)}$, and with that notation, competitive ratio becomes 
$$
\frac{1}{2} bc (bc (bc-t(bc))+t(bc)+3)
$$
(recall here $bc = 9 .. 9.06$). 
}
\end{lemma}

For proving Lemma~\ref{lem: optimizers for c>9/b, speed bound b}, first we recall the known optimizer for the special case $cb=9$ (see Corollary~\ref{cor: solution when c=9/b} within the Proof of Lemma~\ref{lem: optimizers for c>9/b, speed bound b} on page~\pageref{sec: optimizers for c>9/b, speed bound b}), and we identify the tight constraints. Requiring that the exact same inequality constraints to $\nlp{c}{b}$ remain tight, we ask how large can the product $cb$ be so as to have KKT condition hold true. From the corresponding algebraic system, we obtain the answer $cb\leq \gamma_1 \approx 9.06609$. 

Similarly, from Theorem~\ref{thm: optimal speeds naive unbounded speeds} we know the optimizers to $\nlp{c}{b}$ for large enough values of $cb$, and the corresponding tight constraints to the NLP. Again, using KKT conditions, we show that the same constraints remain tight for the optimizers as long as $cb\geq \gamma_2 \approx 11.3414$. This way we obtain the following Lemma. 

\begin{lemma}[Proof on page~\pageref{sec: optimal speeds naive bounded speeds b>=c>=r1}]
\label{lem: optimal speeds naive bounded speeds b>=c>=r1}
For every $c>\rho/b\approx 11.3414/b$, the optimal speeds of \naive{s,r,k} for \ecb{c}{b}
are 
$
s=\sigma/c,  ~
r=\rho/c, ~
k=\kappa/c,
$
i.e. they are the same as for \ecb{c}{\infty}. 
If the target is placed at distance $d$ from the origin, then the induced energy equals 
$584.738 \frac{d}{c^2}$. 
Moreover, the induced competitive ratio is $292.369$, and is independent of $b,c$. 
\end{lemma}

\ignore{
\begin{proof}[Proof of Lemma~\ref{lem: optimal speeds naive bounded speeds b>=c>=r1}]
By Theorem~\ref{thm: optimal speeds naive unbounded speeds}, we know the optimizers to $\nlp{c}{\infty}$; 
$
s=\sigma/c, 
r=\rho/c, 
k=\kappa/c,
$.
These optimizers satisfy the speed bound constraints $s,r,k\leq b$ as long as $\max\{s,r,k\}\leq b$, i.e. $\rho/c\leq b$. Hence, when $c\geq \rho/b$, Non Linear Programs $\nlp{c}{\infty}$, $\nlp{c}{b}$ have the same optimizers. 
\end{proof}
}

The case $\gamma_1 < cb < \gamma_2$ can be solved optimally only numerically, since the best speed values are obtained by roots to a high degree polynomial. Nevertheless, the following lemma proposes a heuristic choice of speeds (that of Theorem~\ref{thm: (sub)optimal speeds naive bounded speeds}) which is surprisingly close to the optimal (as suggested by Theorem~\ref{thm: (sub)optimal speeds naive bounded speeds}, see also Figure~\ref{fig: comparisonVSoptimal-additive}). 

\begin{lemma}[Proof on page~\pageref{sec: naive case between gamma1 gamma2}]
\label{lem: naive case between gamma1 gamma2}
The choices of $s,r,k$ of Theorem~\ref{thm: (sub)optimal speeds naive bounded speeds} when $\gamma_1<cb<\gamma_2$ are feasible. Moreover, the induced competitive ratio is at most 0.03 additively off from the competitive ratio induced by the optimal choices of speeds 
(evaluated numerically). 
\end{lemma}

The trick in order to find ``good enough'' optimizers to $\nlp{c}{b}$ is to guess the subset of inequality constraints that remain tight when $\gamma_1 < cb < \gamma_2$. 
First, we observe that constraint $r\leq b$ is tight for the provable optimizers for all $c,b$ when $cb \in [9, \gamma_1] \cup [\gamma_2, \infty)$. As the only other constraint that switches from being tight to non-tight in the same interval is $k\leq b$, we are motivated to maintain tightness for constraints $r\leq b$ and the time constraint. 
Still the algebraic system associated with the corresponding KKT conditions cannot be solved analytically. To bypass this difficulty, and assuming we know (optimal) speed $s$, we use the tight time constraint to find speed $k$ as a function of $c,b,s$. From numerical calculations, we see that optimal speed $s$ is nearly optimal in $c$, and so we heuristically set $s=\alpha c + \beta$. We choose $\alpha, \beta$ so as to have $s$ satisfy optimality conditions for the boundary values $cb=\gamma_1, \gamma_2$. After we identify all parameters to our solution, we compare the value of our solution to the optimal one (obtained numerically), and we verify (using numerical calculations) that our heuristic solution is only by at most 0.03 additively off. The advantage of our analysis is that we obtain closed formulas for the speed parameters for all values of $cb\geq 9$.


\section{Solving \ecb{c}{b} with Unbounded-Memory Robots}
\label{sec: better algo}

In this section we prove Theorem~\ref{thm: CR non-primitive bounded speeds}, that is we solve \ecb{c}{b} by assuming that the two robots have unbounded memory, and in particular that  they can perform time and state dependent calculations and tasks.  Note that, by scaling, our results hold for all $b,c$ for which $cb=9$. For simplicity our exposition is for the natural case $c=9$ and $b=1$. Also, as before, $d$ will denote the unknown distance of the exit from the origin, still the exit is assumed, for the purposes of performance analysis, to be at least 2 away from the origin. 

Throughout the execution of our evacuation algorithm, robots can be in 3 different states (similar to the case of constant-memory robots). First, both robots start with the \textit{Exploration State} and they remain in this until the exit is located. 
While in the exploration state, robots execute an elaborate exploration that requires synchronous movements in which robots, at a high level, stay in good proximity, still they expand the searched space relatively fast. 
Then, the exit finder enters the \textit{Chasing State} in which the robot, depending on its distance from the origin, calculates a speed, at which to move in order to catch and notify the other robot. 
Lastly, when the two robots meet, they both enter the \textit{Exit State} in which both robots move toward the exit with the smallest possible speed while meeting the time constraint. 

Our algorithm takes as input the values of $c=9,b=1$, and use a speed value $s\leq b$, that will be chosen later. When the exit finder switches its state from Exploration to Chasing, it remembers the distance $d$ of the exit to the origin, as well as the value $k$ of a counter that was used while in the Exploration State. 
When the exit finder catches the other robot, they both switch to the Exit State, and they remember their distance $p$ from the origin, as well as the value of time $t$ that their rendezvous was realized. The speed of their Exit State will be determined as a function of $p,d,t$ (and hence of $s,c,b$ as well).

\subsection{A Critical Component: $l$-Phase Explorations}
\label{sec: kphase explore}

We adopt the language of~\cite{chrobak2015group} in order to discuss a structural property that any feasible evacuation algorithm for \ecb{9}{1} 
satisfies. 
As a result, the purpose of this section is to 
provide high level intuition for our evacuation algorithm that is presented in subsequent sections. 

We refer to the two robots (starting exploration from the origin) as $L$ and $R$, intended to explore to the left and to the right of the origin, respectively. The robot trajectories can be drawn on the Cartesian plane where point-location $(x, -t)$ will correspond to point $x$ on the line being visited by some robot at time $t$. 
The following Theorem is due to~\cite{chrobak2015group} and was originally phrased for the time-evacuation unit-speed robots' problem. We adopt the language of our problem. 
\begin{theorem}
\label{thm: boundary cone}
For any feasible solution to \ecb{9}{1}, the point-location of any robot lies within the cone spanned by vectors $\binom{-1}{-3}, \binom{1}{-3}$. 
\end{theorem}

Next we present some preliminaries toward describing our $k$-phase exploration algorithms.
A \emph{phase} is a pair $(s, r)$ where $s \in [0,1]$ is a speed and $r \in \reals$ is a \textit{distance ratio}, possibly negative.  
An \emph{$l$-phase algorithm} is determined by a position $p_0$ on the line and a sequence $S = (s_1, r_1), \ldots, (s_k, r_l)$ of $l$ phases (movement instructions). 
Whenever $r_ix < 0$, movement will be to the left, whereas $r_ix > 0$ will correspond to movement to the right. 
\begin{algorithm}
       		\textsc{$l$-phase Exploration: given $p_0$ and $S = (s_1, r_1), \ldots, (s_l, r_l)$} \\
		Go to $p_0$ at speed $1/3$ \\
		\Repeat{
			$x \leftarrow $ current position\\
			\For{$i = 1, \ldots, l$}{
				Travel at speed $s_i$ for a distance of $r_i \cdot x$
			}
		}
    \end{algorithm}
We will make sure that each time the loop is executed, position $x$ and corresponding time induce point-locations of the robots that lie in the boundary of the cone of Theorem~\ref{thm: boundary cone}. If a loop starts at location $x$, then it takes additional time $\sum_{i \in [l]} \frac{|r_i| |x|}{s_i}$ to complete one iteration. We will be referring to quantity $1+ \sum_{i \in [l]}\frac{|r_i|}{3s_i}$  as the expansion factor of Exploration $S$.

\subsection{Algorithm $\algo{s}$: The Exploration, Chasing and Exit States}
\label{sec: 3phase explore}

In this section we give a formal description of our evacuation algorithm. The most elaborate part of it is when robots are in Exploration States, in which they will perform $3$-phase exploration. 
It can be shown that $3$-phase exploration based evacuation algorithms that do not violate the constraints of problem \ecb{9}{1} have expansion factor at most 4. Moreover, among those, the ones who minimize the induced energy consumption energy consumption make robots move at speed 1 in the first and third phase\footnote{The proof of these facts is lengthy and technical, and is not required for the correctness of our algorithm, rather it only justifies some parameter choices}. Robot's speed in the second phase will be denoted by $s$. 

We now present a specific $3$-phase exploration algorithm, that we denote by $\algo{s}$, complying with the above conditions, with phases $(-1, 1), (4s/(1-s), s)$ and $(4 - 4s/(1-s), 1)$, where $s$ is an exploration speed to be determined later.  Robot $L$ will execute the 3-phase exploration with starting position -1, while robot $R$ with starting position $2$. 
When subroutine $travel(v, p)$ is invoked, the robot sets its speed to $v$ and, from its current position, goes toward position $p$ on the line until it reaches it. We depict the trajectories of the robots while in the Exploration State in Figure~\ref{fig:cone2018}.
\begin{center}
\fbox{
\begin{minipage}[t]{0.4\textwidth}
  \vspace{0pt}  
            \begin{algorithm}[H]
			\underline{\textsc{Exploration State of $L$}} \\
			$travel(1/3, -1)$ \\
			$k \leftarrow 0$ \\
           	     	\Repeat{
		                	$travel(1, 0)$ \\
		                	$travel(s, -4^{k+1} \cdot \frac{s}{1-s})$ \\
		                	$travel(1, -4^{k+1})$ \\
				$k \leftarrow k + 1$
			}
            \end{algorithm}
\end{minipage}
\begin{minipage}[t]{0.4\textwidth}
  \vspace{0pt}  
            \begin{algorithm}[H]
			\underline{\textsc{Exploration State of $R$}} \\
			$travel(1/3, 2)$ \\
			$k \leftarrow 0$ \\
           	     	\Repeat{
		                	$travel(1, 0)$ \\
		                	$travel(s, 2 \cdot 4^{k+1} \cdot \frac{s}{1-s})$ \\
		                	$travel(1, 2 \cdot 4^{k+1})$ \\
				$k \leftarrow k + 1$
			}
            \end{algorithm}
\end{minipage}
}
\end{center}
A complete execution of one repeat loop within the Exploration State will be referred to as a \textit{round}. 
Variable $k$ counts the number of completed rounds. 
\ignore{
\begin{minipage}[t]{0.30\textwidth}
  \vspace{0pt}  
\begin{figure}[H]
\centering
\includegraphics[width=0.685\linewidth]{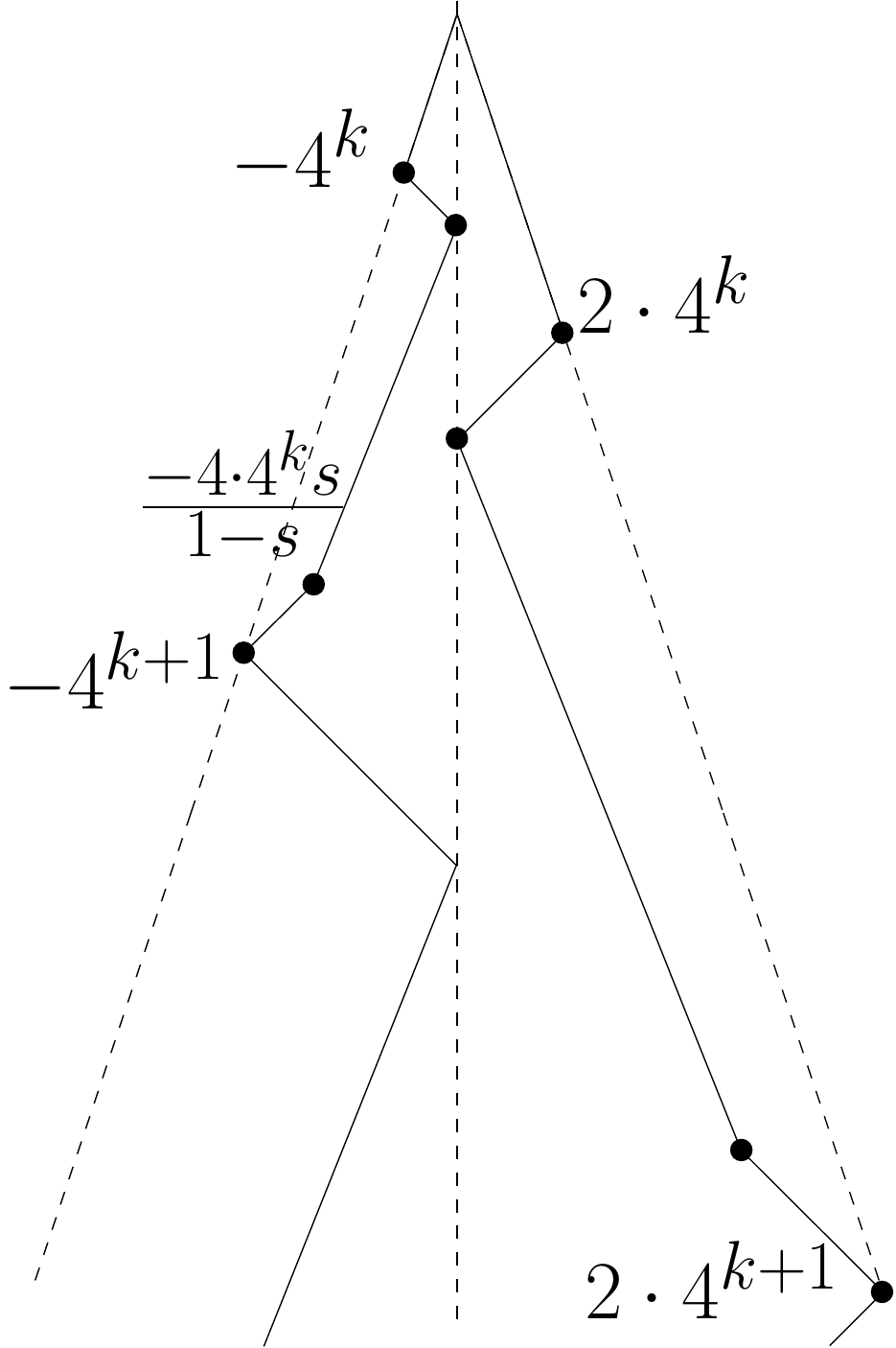}
\caption{A representation of position (x-axis, vertical dashed line is $0$) and time (y-axis), and the trajectory followed by the two robots (solid lines).  The two diagonal dashed lines form the ``1/3 cone'' of Theorem~\ref{thm: boundary cone}.}
\label{fig:cone2018}
\end{figure}
\vline
\end{minipage}
~~~~
\begin{minipage}[t]{0.59\textwidth}
  \vspace{0pt}  
\begin{figure}[H]
\centering
\includegraphics[width=.8\linewidth]{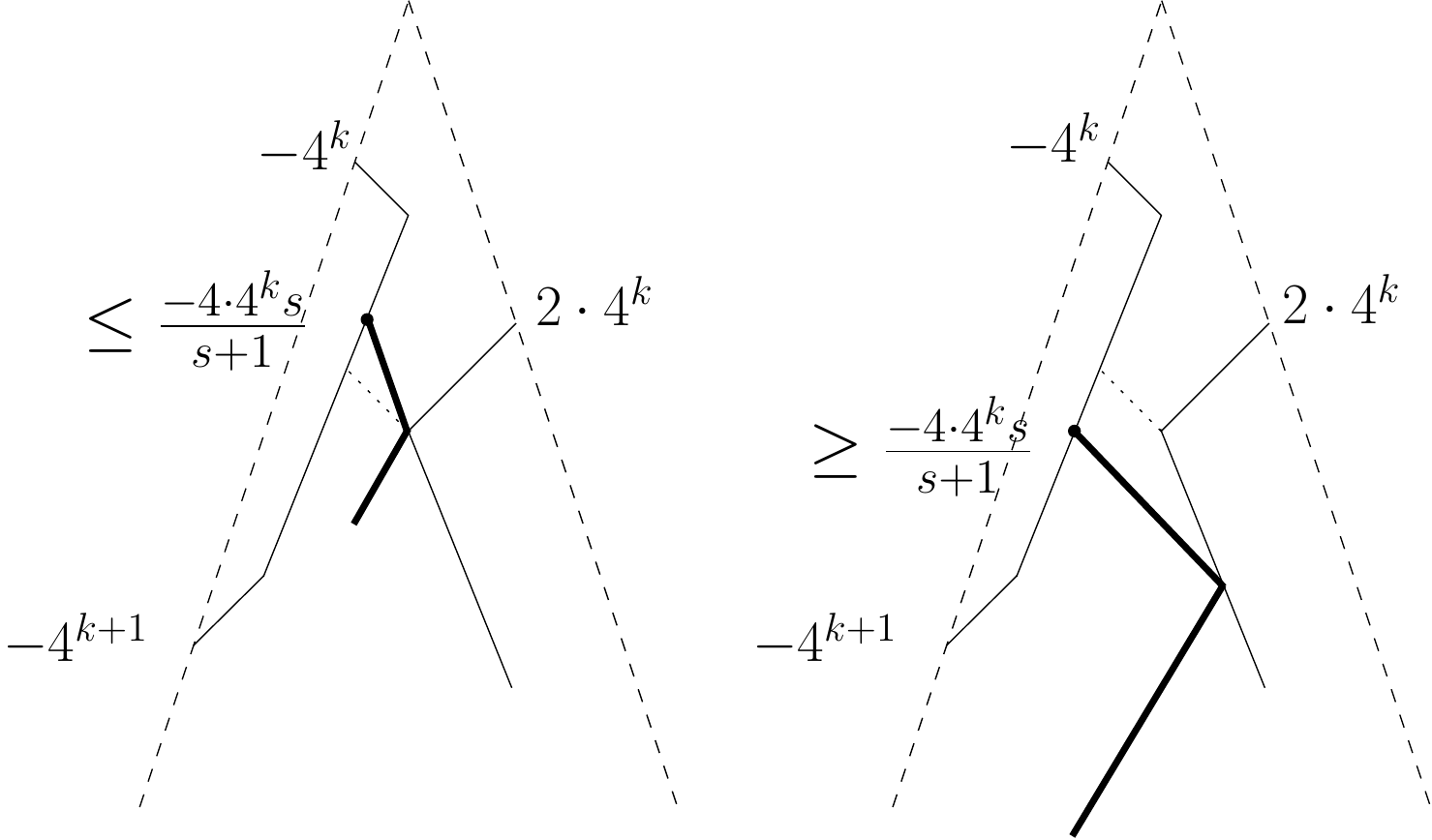}
\caption{The robots' behavior when the exit is found by $L$ is indicated by the bold line. In the first case (left), the catch-up speed is slower than $1$ (and the rendezvous is realized at the turning point of the non-finder), whereas it is $1$ in the second case (right).}
\label{fig:cone-find}
\end{figure}
\end{minipage}
}
Each robot stays in the Exploration State till the exit it found. When switching to the Chasing state (which happens only for the exit finder), robot remembers its current value of counter $k$, as well as the distance $d$ of the exit to the origin. 
Based on these values (as well as $s$) it calculates the most efficient trajectory in order to catch the other robot (predicting, when applicable, that the rendezvous can be realized while the other robot is approaching the exit finder).  
When the rendezvous is realized, robots store their current distance $p$ to the origin, as well as the time $t$ that has already passed. 
Then, robots need to travel distance $p + d$ to reach the exit.  Knowing they have time $9d - t$ remaining, they go to the exit together as slow as possible to reach the exit in time exactly $9d$.
Figure~\ref{fig:cone-find} provides an illustration of the behavior of the robots after finding the exit. 

\begin{center}
\fbox{
\begin{minipage}[t]{0.67\textwidth}
  \vspace{0pt}  
            \begin{algorithm}[H]
			\underline{\textsc{Chasing State}} \\
		$K \leftarrow 4^k$ \\
		\If{I am $R$}{
			$K \leftarrow 2 \cdot 4^k$
		}
		$s' \leftarrow \min
						\left\{
						\dfrac{d}{4K - d/s},1
						\right\}$ 	\\
		Travel toward the other robot at speed $s'$ until meeting it at distance $p$ from the origin, and at time $t$. \\                	
            \end{algorithm}
\end{minipage}
\begin{minipage}[t]{0.28\textwidth}
  \vspace{0pt}  
            \begin{algorithm}[H]
			\underline{\textsc{Exit State}} \\
			$\bar s \leftarrow \frac{p + d}{9d - t}$ \\
			Go toward the exit with speed $\bar s$. 
            \end{algorithm}
\end{minipage}
}
\end{center}

\subsection{Performance Analysis \& an Optimal Choice for Parameter $s$}
\label{sec: analysis}

In this section we are ready to provide the details for proving Theorem~\ref{thm: CR non-primitive bounded speeds}. 
Evacuation algorithm $\algo{s}$ is not feasible to \ecb{c}{b} for all values of speed parameter $s$ (of the Exploration States). We will show later that trajectories induce evacuation time at most $9d$ only if $s\in [1/3,1/2]$. In what follows, and even though we have not fixed the value of $s$ yet, we will assume that $s$ has some value between 1/3 and 1/2. 
The purpose of this section is to fix a value for parameter $s$, show that 
$\algo{s}$ is feasible to \ecb{9}{1}, and subsequently compute the induced energy consumption and competitive ratio. 
As a reminder, each iteration of the repeat loop of the Exploration States is called a round, and $k$ is a counter for these rounds.
\begin{proposition}[Proof on page~\pageref{sec:time-round}]
\label{prop:time-round}
For every $k\geq 0$, and at the start of its $k$-th round, \\
$\bullet$ robot $L$ is at position $-4^{k}$ at time $3 \cdot 4^k$, and \\ 
$\bullet$ robot $R$, is at position $2 \cdot 4^{k}$ at time $6 \cdot 4^k$.
\end{proposition}

\ignore{
\begin{proof}[Proof of Proposition~\ref{prop:time-round}]
The exploration algorithm explicitly ensures that round $k-1$ ends (and hence that round $k$ begins) at the claimed position, so only the time needs to be proved.
Note that the statement is true for both robots when $k = 0$. 
Suppose that when $L$ starts its $(k-1)$-th round, it is at position $-4^{k-1}$ at time
$3 \cdot 4^{k-1}$.  The round ends at time 
$3 \cdot 4^{k-1} + 4^{k-1} + \frac{1}{s} \cdot 4^{k} \cdot \frac{s}{1-s} + (4^k - 4^{k} \cdot \frac{s}{1-s}) = 
3 \cdot 4^k$, and the $k$-th round will start at the claimed time.  The proof is identical for $R$.
\end{proof}
}

Let $X \in \{L, R\}$ be one of the robots.  
We define $K(X, k) = 4^k$ if $X = L$, and $K(X, k) = 2 \cdot 4^k$ if $X = R$, i.e. the position of $X$ at the start of round $k$.
We will often analyze 3 cases for the distance $d$ of the exit with respect to $K := K(X, k)$ (as it also appears in the description of the Chasing State), associated with the following closed intervals
$$
D_1(K) := [K, 4Ks/(s+1)], ~ 
D_2(K) := [4Ks/(s+1), 4Ks/(1-s)], ~
D_3(K) := [4Ks/(1-s), 4K].
$$
We may simply write $D_1, D_2$ and $D_3$ if $K$ is clear from the context.
Note that during the second phase of round $K$, robot $L$ explores $D_1$ and $D_2$,
whereas $D_3$ is explored during the third phase.  The same statement holds for $R$.
The following lemma will be useful in analyzing the worst case evacuation time and energy consumption of our algorithm. 
\begin{lemma}[Proof on page~\pageref{sec:meeting-cases}]
\label{lem:meeting-cases}
Suppose that robot $X \in \{L, R\}$ finds the exit at distance $d$ when its round counter has value $k$.  
Let $p$ and $t$ be, respectively, the position and time at which $X$ first meets with the other robot after having found the exit, 
and set $K := K(X, k)$.  Then the following hold:
\begin{enumerate}
\item
If $d \in D_1$, then $p = 0$ and $t = 8K$.
\item 
If $d \in D_2$, then $|p| = \frac{d + ds - 4Ks}{1-s}$ and $t = 8K + \frac{d + d/s - 4K}{1 - s}$.
\item
If $d \in D_3$, then $|p| = 2ds/(1-s)$ and $t =  8K + 2d + 2ds/(1-s)$.
\end{enumerate}
\end{lemma}

Using the lemma above, we can now prove that $\algo{s}$ meets the speed bound and the evacuation time bound. 
\begin{lemma}[Proof on page~\pageref{sec: algo is feasible}]
\label{lem: algo is feasible}
For any $s \in [1/3, 1/2]$, evacuation algorithm $\algo{s}$ is feasible to \ecb{9}{1}. \end{lemma}

Lemma~\ref{lem:meeting-cases} allows us to derive the speed $s_{b1}, s_{b2}$ and $s_{b3}$ at which both robots go toward the exit after meeting for the cases $d \in D_1, d\in D_2$ and $d \in D_3$, respectively.  We also know the speed $s_{c1}$ at which the exit-finder catches up to the other robot when $d \in D_1$.
We define 
\ignore{
\begin{align*}
s_{b1} &= \frac{d}{9d - 8K}  \quad \quad \quad s_{c1} = \frac{d}{4K - d/s}\\
s_{b2} &= \frac{2d - 4Ks}{d(8 - 9s - 1/s) + 4K(2s-1)} \\
s_{b3} &= \frac{d(1 + s)}{d(7 - 9s) + 8K(s - 1)}
\end{align*}
}
$$
s_{b1} := \tfrac{d}{9d - 8K},~~
s_{c1} = \tfrac{d}{4K - d/s},~~ 
s_{b2} := \tfrac{2d - 4Ks}{d(8 - 9s - 1/s) + 4K(2s-1)},~~ 
s_{b3} := \tfrac{d(1 + s)}{d(7 - 9s) + 8K(s - 1)}
$$
The speed $s_{b2}$ is a simple rearrangement of the speed $\frac{d + qs}{9d - (8K + q)}, \mbox{ where } q = \frac{d + d/s - 4K}{1-s}$, and $s_{b3}$ is obtained by rearranging $\frac{d + 2ds/(1-s)}{9d - (8K + 2d + 2ds/(1-s))}$.

Next compute the energy consumption.
For given $K, d$ and $s$, denote by $E_L(K, d, s)$ the energy spent by robot $L$ from time $3$ to time $9d$ when it exits.
Similarly, $E_R(K, d, s)$ is the energy spent by $R$ from time $6$ to time $9d$.  
Then, then energy consumption is $E(K,d, s) := \tfrac13 + E_L(K,d,s) + E_R(K,d,s)$.
For any $K$ and $s$, we also define 
$F(K,s) := (K - 1)(5 - 4s(s+1))$. 

\begin{lemma}[Proof on page~\pageref{sec:energies}]
\label{lem:energies}
Suppose that robot $X \in \{L, R\}$ finds the exit at distance $d$ when its round counter has value $k$, and let $K := K(X, k)$.  Then 
$$
E(K,d,s) = 
\frac13 + \begin{cases} 
      F(K,s) + 3K + d(s^2 + s_{c1}^2 + 2s_{b1}^2) & \mbox{if } d \in D_1\\
      F(K,s) + 3K + \left( \dfrac{2d - 4Ks}{1-s}\right)(1 + s^2 + 2s_{b2}^2) & \mbox{if } d \in D_2\\
      F(K,s) + 3K - 4Ks(s + 1) + \dfrac{2d}{1-s}(s^3 + s_{b3}^2(s + 1) + 1) & \mbox{if } d \in D_3.
   \end{cases}
$$
\end{lemma}

Denote by $E_i(k,d,s)$ the value of $E(K,d,s)$ when $d\in D_i$, $i=1,2,3$. 
Our intension now is to fix speed value $s$ that solves the following Nonlinear Program 
$$\min_{s\in [1/3,1/2]}
 \left\{ 
	\max
		\left\{
		\sup_{d\in D_1, k\geq 1, X} \frac{E_1(K,d,s)}d,
		\sup_{d\in D_2, k\geq 1, X} \frac{E_2(K,d,s)}d,
		\sup_{d\in D_3, k\geq 1, X} \frac{E_3(K,d,s)}d
		\right\}
\right\}.$$
For every $s\in [1/3,1/2]$ we show in Lemma~\ref{lem: monotonicity of algo s} 
that $\frac{E_1(K,d,s)}d$ is decreasing in $d\in D_1$, 
that $\frac{E_2(K,d,s)}d$ is increasing in $d\in D_2$,
and that $\frac{E_3(K,d,s)}d$ is decreasing in $d\in D_3$. Then, the best parameter $s$ can be chosen so as to make all worst case valued $\frac{E_i(K,d,s)}d$ equal (if possible) when $i=1,2,3$. The optimal $s$ can be found by numerically finding the roots of a high degree polynomial, and accordingly, we heuristically set $s = 0.39403$, inducing the best possible energy consumption for algorithm $\algo{s}$. All relevant formal arguments are within the proof of the next lemma. 


\begin{lemma}[Proof on Page~\pageref{sec: monotonicity of algo s}]
\label{lem: monotonicity of algo s}
On instance $d$ of \ecb{9}{1}, algorithm $\algo{s}$ induces energy consumption at most $8.42588d$, when $s = 0.39403$. 
\end{lemma}

By Lemma~\ref{lem: monotonicity of algo s},  we conclude that for the specific value of $s$, algorithm $\algo{s}$ has competitive ratio
$
\frac{9^2}2 8.42588 \approx 341.24814,
$
concluding the proof of Theorem~\ref{thm: CR non-primitive bounded speeds}.

\begin{theorem}
\label{thm: CR non-primitive bounded speeds}
For every $c,b>0$ with $cb= 9$, there is an evacuation algorithm for unbounded-memory autonomous robots solving \ecb{c}{b} inducing energy consumption $8.42588 b^2 d$ for instances $d$, and competitive ratio 341.24814. 
\end{theorem}

\ignore{
\section{Discussion}
The main contribution of our paper was to introduce an energy
consumption model appropriate to linear search and investigate how the
F2F  communication model affects time/energy trade-offs until
completion of the search by two robots, considering two different computational capabilities for the robots. 
Our approach inspired new
algorithms that take better account of the impact of the change in the
speed of the robots during the course of the search and leads to better
understanding through evaluation of trade-offs of the overall
performance of the algorithm. This raises several interesting problems
worth investigating. In addition to improving the trade-offs in the
algorithms proposed, one may wish to pursue new avenues for research
by examining additional search domains, like the unit disk, in the spirit
of~\cite{CGGKMP}. It would also be natural to consider more realistic
models of search with multiple agents some of which may be
faulty~\cite{isaacCzyzowiczGKKNOS16,PODC16} in such search domains.
}

\section*{Acknowledgements}

Research supported by NSERC discovery grants, NSERC graduate scholarship, and NSF.



\bibliographystyle{plain}
\bibliography{refs}


\appendix

\section{Figures}
\label{sec: omitted figures}

\begin{minipage}[t]{0.39\textwidth}
  \vspace{0pt}  
\begin{figure}[H]
  \centering
  \includegraphics[width=0.8\linewidth]{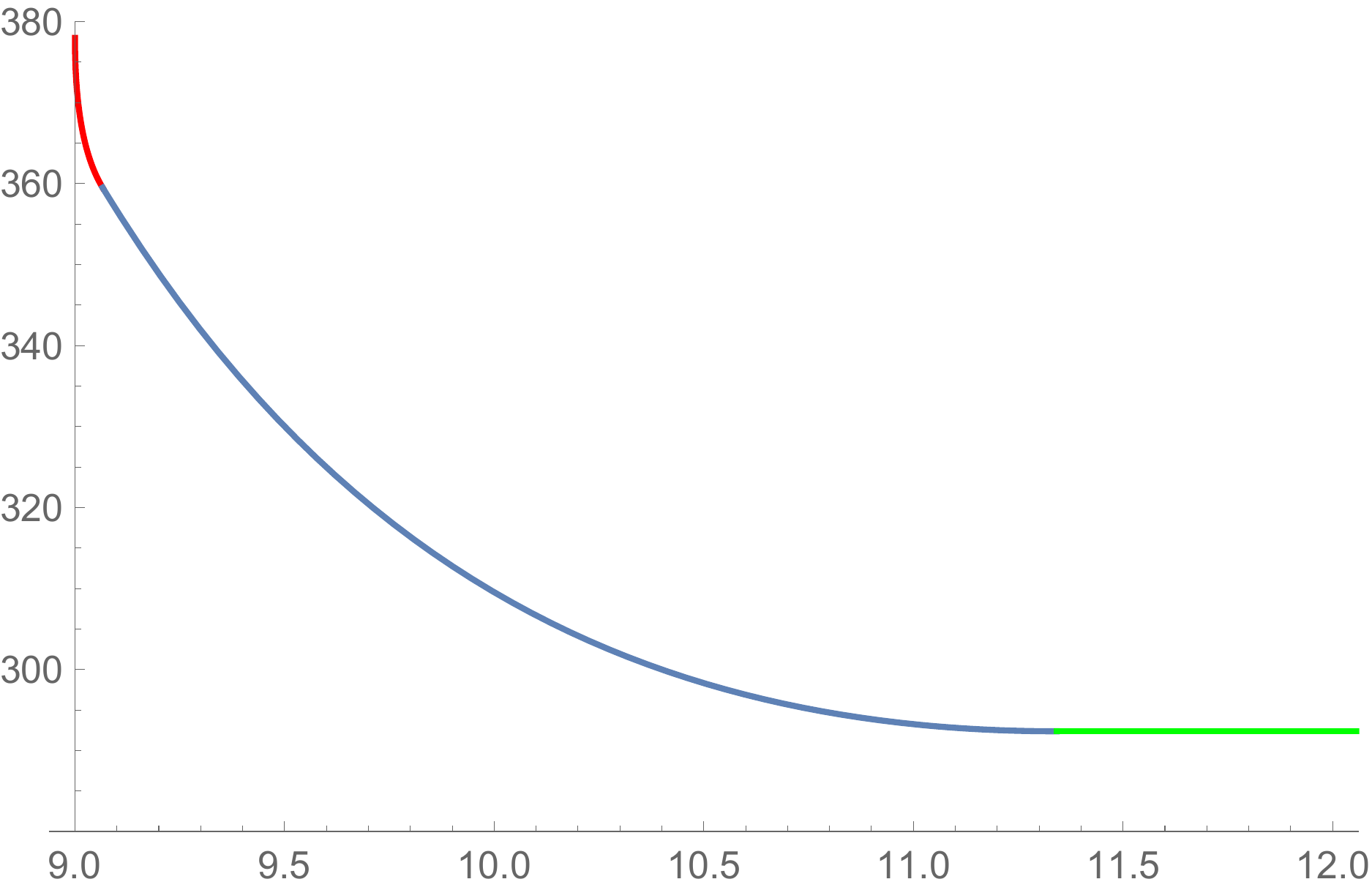}
\caption{
The competitive ratio of algorithm 
\naive{s,r,k} (vertical axis) for the entire spectrum of $cb\geq 9$ (horizontal axis). Red curve corresponds to the case $cb\leq \gamma_1$, blue curve to the case $cb \in (\gamma_1,\gamma_2)$ and green curve to the case $cb\geq \gamma_2$. The curve is continuous and differentiable for all $cb\geq 9$. 
}
\label{fig: compratiof2f}
\end{figure}
\vline
\end{minipage}
~~~~~
\begin{minipage}[t]{0.54\textwidth}
  \vspace{0pt}  
\begin{figure}[H]
  \centering
  \includegraphics[width=0.65\linewidth]{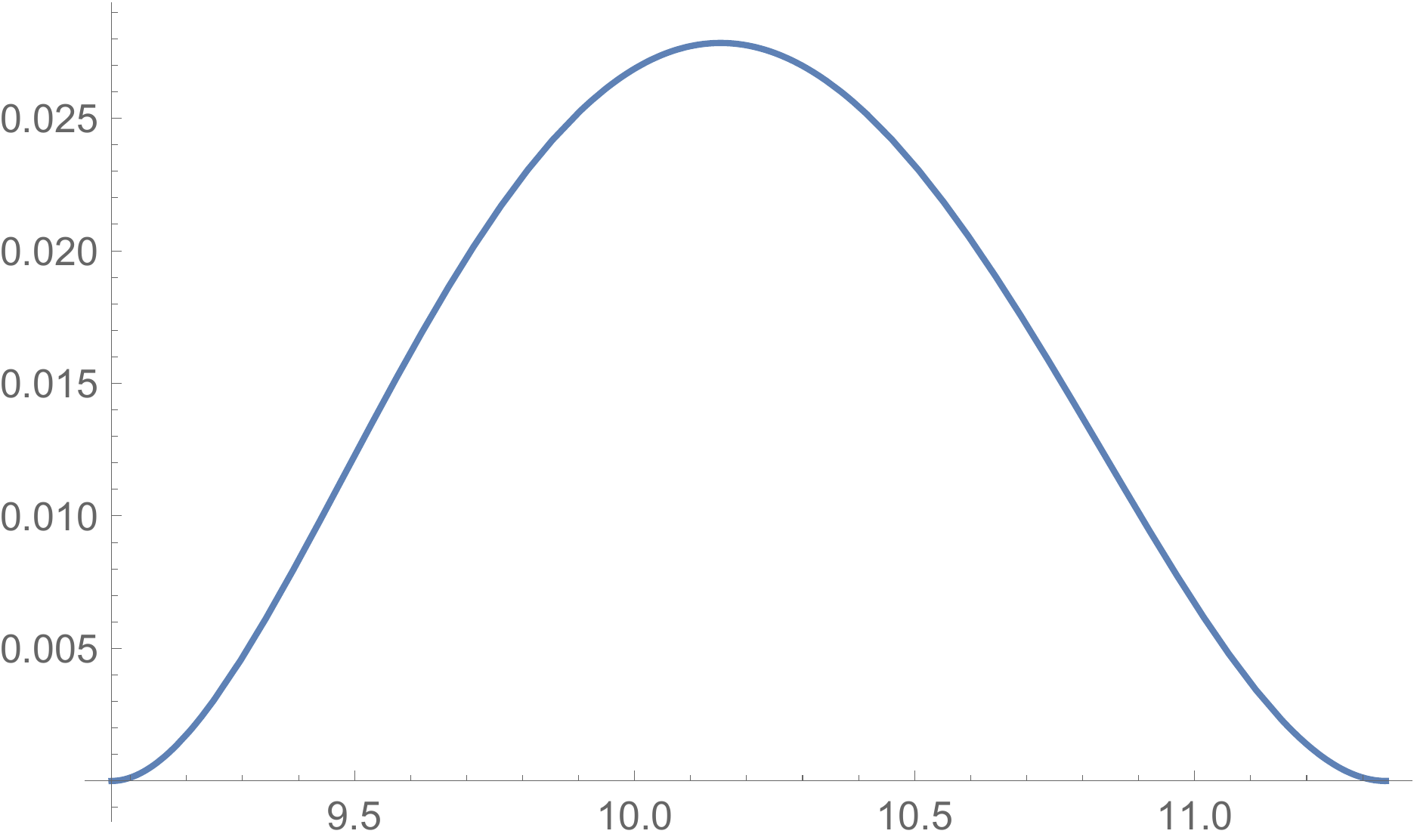}
\caption{
Comparison between the competitive ratio achieved by using \textit{optimal} speed parameters to $\nlp{c}{b}$ of Theorem~\ref{thm: naive as NLP} (calculated numerically using software) and the competitive ratio achieved by the choices of Theorem~\ref{thm: (sub)optimal speeds naive bounded speeds}. 
The vertical axis is the difference of the competitive ratios, and the horizontal axis corresponds to the values of $cb\in (\gamma_1,\gamma_2)$ (for all other values of $cb$ the difference is provably 0). 
}
\label{fig: comparisonVSoptimal-additive}
\end{figure}
\end{minipage}

\begin{minipage}[t]{0.30\textwidth}
  \vspace{0pt}  
\begin{figure}[H]
\centering
\includegraphics[width=0.685\linewidth]{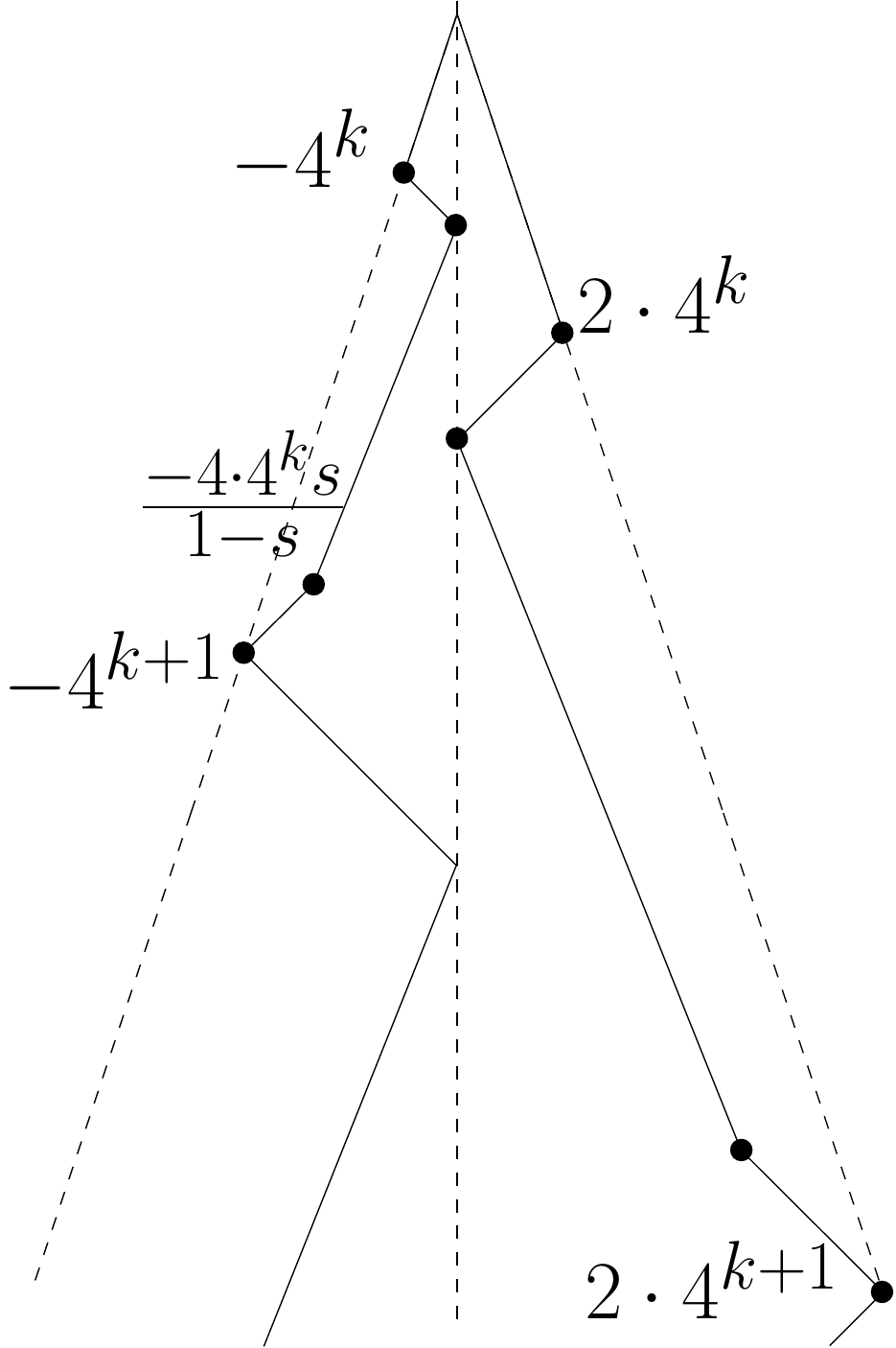}
\caption{A representation of position (x-axis, vertical dashed line is $0$) and time (y-axis), and the trajectory followed by the two robots (solid lines).  The two diagonal dashed lines form the ``1/3 cone'' of Theorem~\ref{thm: boundary cone}.}
\label{fig:cone2018}
\end{figure}
\vline
\end{minipage}
~~~~
\begin{minipage}[t]{0.59\textwidth}
  \vspace{0pt}  
\begin{figure}[H]
\centering
\includegraphics[width=.8\linewidth]{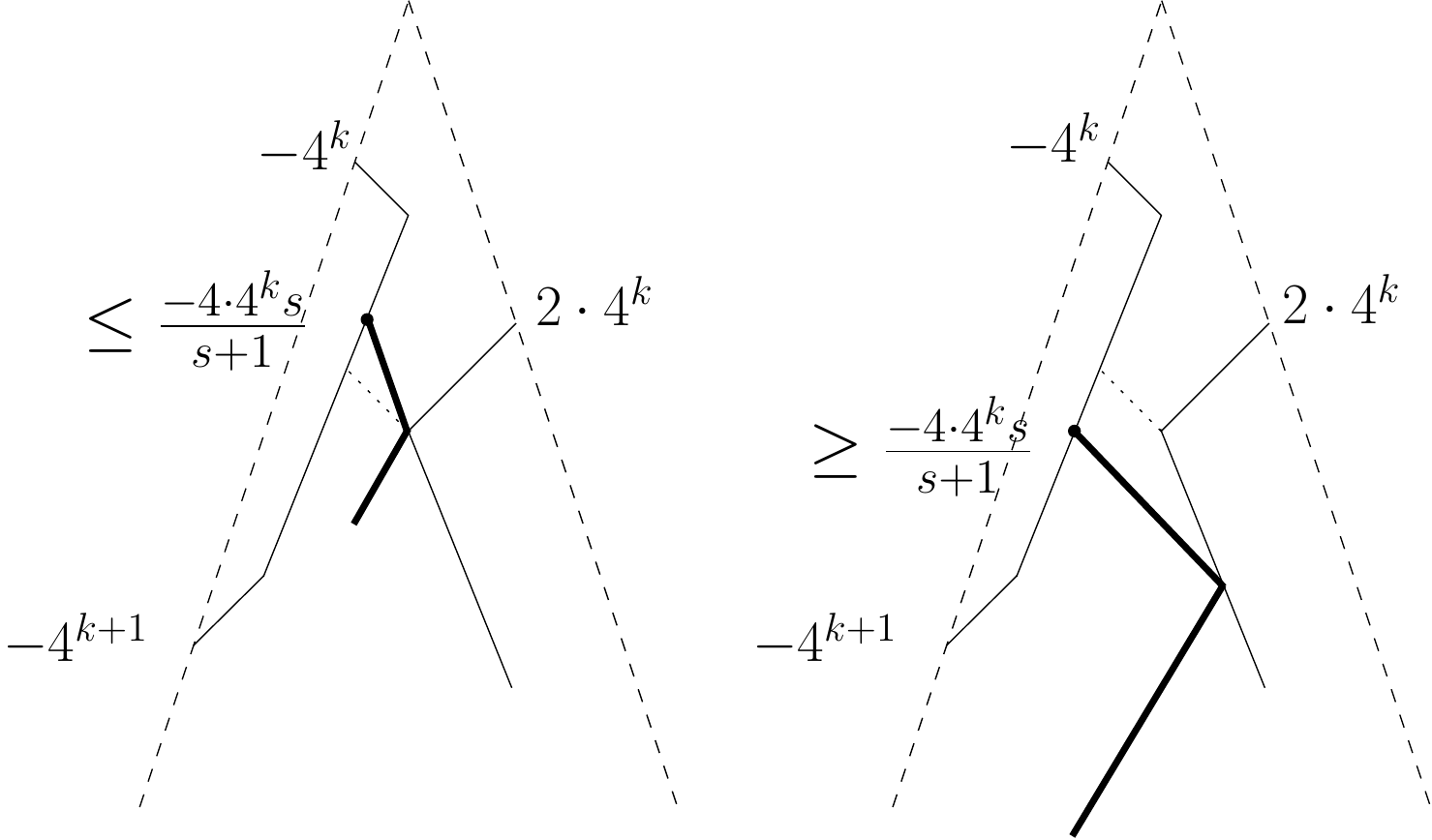}
\caption{The robots' behavior when the exit is found by $L$ is indicated by the bold line. In the first case (left), the catch-up speed is slower than $1$ (and the rendezvous is realized at the turning point of the non-finder), whereas it is $1$ in the second case (right).}
\label{fig:cone-find}
\end{figure}
\end{minipage}

\section{Observation~\ref{obs: opt offline bounded speeds}}
\label{sec: opt offline bounded speeds}

\begin{observation}
\label{obs: opt offline bounded speeds}
\ecb{c}{b} is well defined for each $b,c>0$ with $bc\geq1$, and the optimal solution, given that instance $d$ is known, equals $\frac{2d}{c^2}$. 
\end{observation}

\begin{proof}[Proof of Observation~\ref{obs: opt offline bounded speeds}]
Given that the location of the exit is known, and by symmetry, it is immediate that both robots have the same optimal speed, call it $s$, and they move in the direction of the exit. The induced evacuation time is then $d/s$, and the induced evacuation 
energy is $2d\cdot s^2$. For a feasible solution we require that $d/s\leq c\cdot d$ and that $0<s\leq b$, and hence, the optimal offline solution is obtained as the solution to 
$\min_s \{ s^2: ~1/c \leq s\leq b \}.$
For a feasible solution to exist, we need $bc\geq1$. Moreover, it is immediate that the optimal choice is $s=1/c$, inducing energy consumption $2d\cdot s^2=2d/c^2$.
\end{proof}

\section{Proofs Omitted from Section~\ref{sec: naive algo}.}

\subsection{Lemma~\ref{lem: energy time formulas}}
\label{sec: energy time formulas}
\begin{proof}[Proof of Lemma~\ref{lem: energy time formulas}]
Consider the moment that the exit is located, after time $d/s$ time of searching. 
The robot that now chases the other speed-$s$ robot at constant speed $r>s$ will reach it after $2d/(r-s)$ time. To see this note that the configuration is equivalent to that the speed-$s$ robot is immobile and the other robot moves at speed $r-s$, having to traverse a total distance of $2d$. Moreover, the speed-$s$ robot traverses an additional length $2ds/(r-s)$ segment till it is caught, being a total of $2ds/(r-s)+2d$ away from the exit. Once robots meet, the walk to the exit at speed $k$, which takes additional time $(2ds/(r-s)+2d)/k$. Overall the evacuation time equals
$$
\frac{d}{s}+\frac{2d}{r-s}+\frac{2ds/(r-s)+2d}{k} =
d\left( \frac{2 (k+r)}{k (r-s)}+\frac{1}{s}\right).
$$

Similarly we compute the total energy till both robots reach the exit. The energy spent by the finder is
$$
d\cdot s^2 + \left( \frac{2ds}{r-s} + 2d \right) \cdot r^2 + \left( \frac{2ds}{r-s} + 2d \right) \cdot k^2,
$$
while the energy spent by the non finder is 
$$
\left(d + \frac{2ds}{r-s}\right) \cdot s^2  + \left( \frac{2ds}{r-s} + 2d \right) \cdot k^2.
$$
Adding the two quantities and simplifying gives the promised formula. 
\end{proof}

\subsection{Theorem~\ref{thm: naive as NLP}}
\label{sec: naive as NLP}
\begin{proof}[Proof of Theorem~\ref{thm: naive as NLP}]
Note that in $\nlp{c}{\infty}$ that aims to provide a solution to \ec{c}, constraints $s,r,k\leq \infty$ are simply omitted. 
In particular, the theorem above claims that when $b=\infty$, i.e. when speeds are unbounded, algorithm \naive{s,r,k} always admits some feasible solution. 
In what follows, we prove all claims of the theorem. 

By Lemma~\ref{lem: energy time formulas}, the energy performance of \naive{s,r,k} equals $2d\cdot \eenergy{s,r,k}$, and the induced evacuation time is $d\cdot \ttime{s,r,k}$. For the values of $s,r,k$ to be feasible, we need that 
$0< s,r,k \leq b$, 
that 
$
r>k
$
and that 
$
d\cdot \ttime{s,r,k}\leq cd
$. 
Clearly the latter time constraint simplifies to the time constraint of $\nlp{c}{b}$, while the objective value can be scaled by $d>0$ without affecting the optimizers to the NLP, if such optimizers exist. 
Finally note that even though the strict inequalities become non strict inequalities in the NLP, speeds evaluations for which any of $s,r,k$ is 0 or $r=k$ violates the time constraint (for any fixed $c>0$). Therefore, $\nlp{c}{b}$ correctly formulates the problem of choosing optimal values for \naive{s,r,k} for solving \ecb{c}{b}. 

The next two lemmas show that the Naive algorithm can solve problem \ecb{c}{b} for the entire spectrum of $c,b$ values for which the problem admits solutions, as per Lemma~\ref{lem: restrictions on cb f2f}. 

\begin{lemma}
\label{lem: optimizer exists}
For every $c$, problem $\nlp{c}{\infty}$ admits an optimal solution. 
\end{lemma}

\begin{proof}[Proof of Lemma~\ref{lem: optimizer exists}]
Consider the redundant constraints $s,r,k\geq 1/c$ that can be derived by the existing constraints of~$\nlp{c}{\infty}$ (note that if all speeds are not at least $1/c$ then clearly the time constraint is violated). For the same reason, it is also easy to see that $r-s\geq 1/c$, since again we would have a violation of the time constraint. 

Next, it is easy to check that $s=7/c, r=14/c, k=7/c$ is a feasible solution, hence the NLP is not infeasible. 
The value of the objective for this evaluation is $686/c^2$. But then, notice that the objective is bounded from below by $s^2+r^2+2k^2$. Hence, if an optimal solution exists, constraints 
$s,r,k\leq \sqrt{686}/c$ are valid for the optimizers. We may add these constraints to~$\nlp{c}{\infty}$, resulting into a compact (closed and bounded) feasible region. But then, note that the objective is continuous over the new compact feasible region, hence from the extreme value theorem it attains a minimum. 
\end{proof}


\begin{lemma}
\label{lem: bounded speeds spectrum of r}
There exist $s,r,k$ for which \naive{s,r,k} induces a feasible solution to \ecb{c}{b}\ if and only if $c\geq 9/b$. 
\end{lemma}

\begin{proof}[Proof of Lemma~\ref{lem: bounded speeds spectrum of r}]
Consider the problem of minimizing completion time of the Naive Algorithm, given that the speeds are all bounded above by $b$. 
The corresponding NLP that solves the problem reads as. 
\begin{align}
\min & \frac{2 (k+r)}{k (r-s)}+\frac{1}{s} \label{min time, bounded speeds} \\
s.t. ~~&
r \geq s \notag \\
&0 \leq s,r,k\leq b \notag
\end{align}
Note that it is enough to prove that the optimal value to \eqref{min time, bounded speeds} is $9/b$. 
Indeed, that would imply that no speeds exist that induce completion time less than $9/b$, making the corresponding feasible region of~$\nlp{c}{b}$ empty if $c<9/b$. 

Now we show that the optimal value to \eqref{min time, bounded speeds} is $9/b$, by showing that the unique optimizers to the NLP are $r=k=b$ and $s=b/3$. Indeed, note that 
$$
\frac{\partial}{\partial r} \left( \frac{2 (k+r)}{k (r-s)}+\frac{1}{s} \right)
= -\frac{2 (k+s)}{k (r-s)^2}
$$
which is strictly negative for all feasible $s,r,k$ with $r>s$. Hence, there is no optimal solution for which $r<b$, as otherwise by increasing $r$ one could improve the value of the objective. Similarly we observe that 
$$
\frac{\partial}{\partial k} \left( \frac{2 (k+r)}{k (r-s)}+\frac{1}{s} \right)
= -\frac{2 r}{k^2 (r-s)}
$$
which is again strictly negative for all feasible $s,r,k$ with $r>s$. Hence, there is no optimal solution for which $k<b$, as otherwise by increasing $k$ one could improve the value of the objective.

To conclude, in an optimal solution to~\eqref{min time, bounded speeds} we have that $r=k=b$, and hence one needs to find $s$ minimizing $g(s,b,b)=\frac{4}{b-s}+\frac{1}{s}$. For this we compute 
$$
\frac{\partial}{\partial s} g(s,b,b) = \frac{4}{(b-s)^2}-\frac{1}{s^2}
$$
and it is easy to see that $g(s,b,b)=0$ if and only if $s=b/3$ or $s=-b$ (and the latter is infeasible). At the same time, $g(s,b,b)$ is convex when $s\leq b$ because $\frac{\partial^2}{\partial s^2} g(s,b,b) = \frac{8}{(b-s)^3}+\frac{2}{s^3} >0$, hence $s=b/3$ corresponds to the unique minimizer. 
\end{proof}


The last component of Theorem~\ref{thm: naive as NLP} that requires justification pertains to the competitive ratio. 
Now fix $b,c>0$ for which $bc\geq 9$, and let $\eenergy{s_0,r_0,k_0}$ be the optimal solution to $\nlp{c}{b}$ (corresponding to the optimal choices of algorithm \naive{s,r,k}). By Lemma~\ref{lem: energy time formulas} the induced energy consumption is $2d\cdot \eenergy{s_0,r_0,k_0}$. Then,  the competitive ratio of the algorithm is 
$
\sup_{d>0} \frac{c^2}{2d}~ 2d\cdot \eenergy{s_0,r_0,k_0} = c^2 \cdot \eenergy{s_0,r_0,k_0}.
$
\end{proof}

\subsection{Theorem~\ref{thm: optimal speeds naive unbounded speeds}}
\label{sec: optimal speeds naive unbounded speeds}
\begin{proof}[Proof of Theorem~\ref{thm: optimal speeds naive unbounded speeds}]
First we observe that $s=\frac{\sigma}{c}, r=\frac{\rho}{c}, k=\frac{\kappa}{c}$ are indeed feasible to $\nlp{c}{\infty}$ (for every $c>0$), since
$$
\ttime{\frac{\sigma}{c}, \frac{\rho}{c}, \frac{\kappa}{c}}
=c \left(\frac{2 (\kappa +\rho )}{\kappa  (\rho -\sigma )}+\frac{1}{\sigma }\right)
$$
and in particular, for the values of $\sigma, \rho, \kappa$ described above we have 
$\left(\frac{2 (\kappa +\rho )}{\kappa  (\rho -\sigma )}+\frac{1}{\sigma }\right) \approx 1$ 
(from the formal definition of $\sigma, \rho, \kappa$ that appears later, it will be clear that expression will be exactly equal to 1). 
Moreover, by Theorem~\ref{thm: naive as NLP}, the competitive ratio of \naive{\frac{\sigma}{c}, \frac{\rho}{c}, \frac{\kappa}{c}} is 
$$
c^2 \cdot \eenergy{\frac{\sigma}{c}, \frac{\rho}{c}, \frac{\kappa}{c}} = \frac{\rho  \left(2 \kappa ^2+\rho ^2+\sigma ^2\right)}{\rho -\sigma },
$$
as claimed. 

In the remaining of the section we prove that the choices for $s,r,k$ of Theorem~\ref{thm: optimal speeds naive unbounded speeds} are indeed optimal for \naive{s,r,k}. 
First we establish a structural property of optimal speeds choices for \naive{s,r,k}. 

\begin{lemma}
\label{lem: tight constraint}
For any $c>0$, optimal solutions  to $\nlp{c}{\infty}$ 
satisfy constraint $\ttime{s,r,k} \leq c$ tightly.
\end{lemma}

\begin{proof}
Consider an optimal solution $\bar s, \bar r, \bar k$. As noted before, we must have 
$\bar s, \bar r, \bar k>0$ and $\bar r>\bar s$, as otherwise the values would be infeasible. 

Next note that the time constraint can be rewritten as 
$$
k\geq  \frac{2 r s}{c r s-c s^2-r-s}
$$
For the sake of contradiction, assume that the time constraint is not tight for $\bar s, \bar r, \bar k$. Then, there is $\epsilon>0$ so that $\bar s, \bar r, k'$ is a feasible solution, where $k'=\bar k - \epsilon>0$. But then, the objective value strictly decreases, a contradiction to optimality. 
\end{proof}

We will soon derive the optimizers to $\nlp{c}{\infty}$ using Karush-Kuhn-Tucker (KKT) conditions. Before that, we observe that solutions are scalable with respect to $c$, which will also allow us to simplify our calculations. 

\begin{lemma}
\label{lem: normalized reduction}
Let $s', r', k'$ be the optimizers to~$\nlp{1}{\infty}$ inducing optimal energy $E$. 
Then, for any $c$, the optimizers to~$\nlp{c}{\infty}$ are $\bar s = s'/c, \bar r=r'/c, \bar k = k'/c$, and the induced optimal energy is $\frac1{c^2}E$. 
\end{lemma}

\begin{proof}
Note that the triplet $(s,r,k)$ is feasible to~$\nlp{c}{\infty}$ (for a specific $c$) 
if and only if 
the triplet $(c\cdot s,c \cdot r,c \cdot k)$ is feasible to~$\nlp{1}{\infty}$.
Moreover, it is straightforward that when speeds are scaled by $c$, the induced energy is scaled by $c^2$. 
Hence, for every $c>0$ there is a bijection between feasible (and optimal) solutions to ~$\nlp{c}{\infty}$ and ~$\nlp{1}{\infty}$.
\end{proof}

We are therefore motivated to solve $\nlp{1}{\infty}$, and that will allow us to derive the optimizers for~$\nlp{c}{\infty}$, for any $c>0$.

\begin{lemma}
\label{lem:optimizers of pcprime}
The optimal solution to~$\nlp{1}{\infty}$ is obtained for 
$$
s = \sigma \approx  2.65976,~
r = \rho \approx 11.3414 ,~
k=\kappa \approx 6.63709
$$
and the optimal NLP value is $\frac{\rho  \left(2 \kappa ^2+\rho ^2+\sigma ^2\right)}{\rho -\sigma } \approx 292.37$.
\end{lemma}

\begin{proof}
By KKT conditions, we know that, necessarily, all minimizers of $\eenergy{s,r,k}$ satisfy the condition that 
$-\nabla \eenergy{s,r,k}$ is a conical combination of tight constraints (for the optimizers). Lemma~\ref{lem: tight constraint}
asserts that $\ttime{s,r,k} = 1$ has to be satisfied for all optimizers $s,r,k$.
At the same time, recall that, by the proof of Lemma~\ref{lem: optimizer exists}, none of the constraints  $r\geq s$ and $s,r,k\geq 0$ can be tight for an optimizer. Hence, KKT conditions imply that any optimizer $s,r,k$ satisfies, necessarily,  the following system of nonlinear constraints
\begin{align*}
-\nabla \eenergy{s,r,k} &= \lambda \nabla \ttime{s,r,k} \\
\ttime{s,r,k}&=1 \\
\lambda &\geq 0
\end{align*}
More explicitly, the first equality constraints is
\begin{align*}
\left(
\begin{array}{c}
r-\frac{2 r \left(k^2+r^2\right)}{(r-s)^2}\\
\frac{2 k^2 s-2 r^3+3 r^2 s+s^3}{(r-s)^2}\\
\frac{4 k r}{s-r}
\end{array}
\right)
&=
\lambda \left(
\begin{array}{c}
\frac{2 (k+r)}{k (r-s)^2}-\frac{1}{s^2}\\
-\frac{2 (k+s)}{k (r-s)^2}\\
\frac{2 r}{k^2 (s-r)}
\end{array}
\right)
\end{align*}
From the 3rd coordinates of the gradients, we obtain that $\lambda=2k^3$, which directly implies that the dual multiplier $\lambda$ preserves the correct sign for the necessary optimality conditions. 

Hence, the original system of nonlinear constraints is equivalent to that 
\begin{align*}
r-\frac{2 r \left(k^2+r^2\right)}{(r-s)^2}
&=2k^3
\left(
\frac{2 (k+r)}{k (r-s)^2}-\frac{1}{s^2}
\right) \\
-2 k^2 s+2 r^3-3 r^2 s-s^3
&= 4k^2(k+s) \\
\frac{2 (k+r)}{k (r-s)}+\frac{1}{s} &=1
\end{align*}
Using software numerical methods, we see that the above algebraic system admits the following 3 real roots for $(s,r,k)$: 
\begin{align*}
(2.659764883844293,~11.341425445393606,~6.637089776204052) & (\textrm{multiplicity 1}) \\
(-0.6115006613361799,~0.47813267995355124,~1.0972211311317337) & (\textrm{multiplicity 2})
\end{align*}

Since also all speeds are nonnegative, we obtain the unique candidate optimizer 
$$(\sigma,\rho,\kappa)=(2.65976, 11.3414, 6.63709).$$ 
To verify that indeed $(\sigma,\rho,\kappa)$ is a minimizer, we compute
$$
\nabla^2 \eenergy{s,r,k}=
\left(
\begin{array}{ccc}
 \frac{4 r \left(k^2+r^2\right)}{(r-s)^3} & \frac{(r+s) \left(-2 k^2+r^2+s^2-4 r s\right)}{(r-s)^3} & \frac{4 k r}{(r-s)^2} \\
 \frac{(r+s) \left(-2 k^2+r^2+s^2-4 r s\right)}{(r-s)^3} & \frac{2 \left(r^3-3 s r^2+3 s^2 r+s^3+2 k^2 s\right)}{(r-s)^3} & -\frac{4 k s}{(r-s)^2} \\
 \frac{4 k r}{(r-s)^2} & -\frac{4 k s}{(r-s)^2} & \frac{4 r}{r-s} \\
\end{array}
\right).
$$
\ignore{
MatrixForm[
 FullSimplify[
  D[f[s, r, k], {{s, r, k}, 2}]
  ]
 ]
}
Moreover, 
$$
\nabla^2 \eenergy{\sigma,\rho,\kappa}=
\left(
\begin{array}{ccc}
 11.9718 & -1.56333 & 3.99485 \\
 -1.56333 & 2.83125 & -0.936864 \\
 3.99485 & -0.936864 & 5.22546 \\
\end{array}
\right)
$$
which has eigenvalues $14.1183, 3.41098, 2.49927$, hence it is PSD. 
As a result, $f(s,r,k)$ is locally convex at $(\sigma,\rho,\kappa)$, and therefore $(\sigma,\rho,\kappa)$ is a local minimizer to~$\nlp{1}{\infty}$. 
As we showed earlier, $(\sigma,\rho,\kappa)$ is the only candidate optimizer, hence a global minimizer as well. 

\ignore{
s0 = 2.65976;
r0 = 11.3414;
k0 = 6.63709;
hessian[s_, r_, 
  k_] := {{(2 r)/(r - s) + (4 r s)/(r - s)^2 + (
    2 r (2 k^2 + r^2 + s^2))/(r - s)^3, (2 r^2)/(r - s)^2 - (
    2 r s)/(r - s)^2 + (2 s)/(r - s) - (
    2 r (2 k^2 + r^2 + s^2))/(r - s)^3 + (
    2 k^2 + r^2 + s^2)/(r - s)^2, (
   4 k r)/(r - s)^2}, {(2 r^2)/(r - s)^2 - (2 r s)/(r - s)^2 + (2 s)/(
    r - s) - (2 r (2 k^2 + r^2 + s^2))/(r - s)^3 + (
    2 k^2 + r^2 + s^2)/(r - s)^2, -((4 r^2)/(r - s)^2) + (6 r)/(
    r - s) + (2 r (2 k^2 + r^2 + s^2))/(r - s)^3 - (
    2 (2 k^2 + r^2 + s^2))/(r - s)^2, -((4 k r)/(r - s)^2) + (4 k)/(
    r - s)}, {(
   4 k r)/(r - s)^2, -((4 k r)/(r - s)^2) + (4 k)/(r - s), (4 r)/(
   r - s)}}
   MatrixForm[ hessian[s0, r0, k0] ]
   }

\ignore{
c1[s_, r_, k_] := 
  r - (2 r (k^2 + r^2))/(r - s)^2 - 
   2*k^3*((2 (k + r))/(k (r - s)^2) - 1/s^2);
(* c2[s_,r_,k_]:= (-2 r^3+2 k^2 s+3 r^2 s+s^3)/(r-s)^2-2*k^3*(-((2 \
(k+s))/(k (r-s)^2))) ;*)

 c2[s_, r_, k_] := (-2 r^3 + 2 k^2 s + 3 r^2 s + s^3)/1 + 
   2*k^2*((2 (k + s))/1 ) ;

c3[s_, r_, k_] := 2*(k + r)/k/(r - s) + 1/s - 1;

NSolve[ 
 {c1[s, r, k] == 0, c2[s, r, k] == 0, c3[s, r, k] == 0}, {s, r, k}]
}

\end{proof}

\ignore{
f[s_, r_, k_] := r/(r - s)*(s^2 + r^2 + 2*k^2);
g[s_, r_, k_] := 2*(k + r)/k/(r - s) + 1/s;

MatrixForm[
 FullSimplify[ -Grad[f[s, r, k], {s, r, k}] ]
 ]
}

\ignore{
time[s_, r_, k_] := (2 (k + r))/(k (r - s)) + 1/s;
energy[s_, r_, k_] := ( r (2 k^2 + r^2 + s^2))/(r - s)
Minimize[ {energy[s, r, k], 
  time[s, r, k] == 1. && s >= 0 && r >= 0 && k >= 0}, {s, r, k}] 
}

Lemma~\ref{lem:optimizers of pcprime} together with Lemma~\ref{lem: normalized reduction} imply that for any $c>0$ the optimal solution to $\nlp{c}{\infty}$ is exactly for 
$$
s=\frac{\sigma}{c}, ~r=\frac{\rho}{c}, ~k=\frac{\kappa}{c}
$$
and hence, the proof of Theorem~\ref{thm: optimal speeds naive unbounded speeds} follows. 
\end{proof}

\subsection{Lemma~\ref{lem: optimizers for c>9/b, speed bound b}}
\label{sec: optimizers for c>9/b, speed bound b}
\begin{proof}[Proof of Lemma~\ref{lem: optimizers for c>9/b, speed bound b}]

An immediate corollary from the proof of Lemma~\ref{lem: bounded speeds spectrum of r} (within the proof of Theorem~\ref{thm: naive as NLP}) is the following
\begin{corollary}
\label{cor: solution when c=9/b}
The unique solution to~$\nlp{c}{b}$ when $c=9/b$ is given by 
$$
r=k=b, s=\frac{b}3,
$$
inducing energy $\frac{28 b^2d}{3}$, and competitive ratio $\frac{14 (cb)^2}{3}=378$.
\end{corollary}

Next we find solutions for $c>9/b$ so that $r,k\leq b$ remain tight.
Since, when $c=9/b$, there is only one optimizer $s=b/3, r=b,k=b$, two inequality constraints are tight. The next calculations investigate the spectrum of $c$ for which the same constraints remain tight for the optimizer. 

We write 1st order necessary optimality conditions for $\nlp{c}{b}$, given that the candidate optimizer satisfies the time constraint, and the two $r,k\leq b$ speed bound constraints tightly
\ignore{
\begin{align}
& \frac{2 (k+r)}{k (r-s)}+\frac{1}{s} \leq c 		\label{time c bound b} \\
& r\leq b															\label{r bound bound b}\\
& k\leq b															\label{k bound bound b}
\end{align}
imply that 
}
\begin{align*}
-\nabla \eenergy{s,r,k} & = 
\lambda_1
\nabla \ttime{s,r,k}
+
\lambda_2 
\left(
\begin{array}{c}
0\\
1\\
0
\end{array}
\right) 
+
\lambda_3 
\left(
\begin{array}{c}
0\\
0\\
1
\end{array}
\right) 
 \\
\ttime{s,r,k}&=c\\
r&=b \\
k&=b\\
\lambda_1,\lambda_2,\lambda_3 &\geq 0
\end{align*}

From the tight time constraint, and solving for $s$ we obtain that 
$$
s_{1,2}=\frac{\pm \sqrt{b^2 c^2-10 b c+9}+b c-3}{2 c}
$$

For each $s\in \{s_1, s_2\}$, the first gradient equality defines a linear system over $\lambda_1, \lambda_2, \lambda_3$ whose solutions are
$$
\lambda_1 =\frac{(s-3) s^2}{3 s-1}, ~~
\lambda_2 = -\frac{-5 s^3-9 s+2}{(s-1) (3 s-1)}, ~~
\lambda_3=-\frac{2 \left(s^3-3 s^2-6 s+2\right)}{(s-1) (3 s-1)}.
$$
$$
\lambda_1 =\frac{b s^2 (3 b-s)}{b-3 s}, ~~
\lambda_2 =-\frac{2 b^3-9 b^2 s-5 s^3}{(b-3 s) (b-s)}, ~~
\lambda_3=-\frac{2 \left(2 b^3-6 b^2 s-3 b s^2+s^3\right)}{(b-3 s) (b-s)}
$$
respectively. 
As long as all dual multiplies $\lambda_i=\lambda_i(s)$ are positive, corresponding solution $(s,b,b)$ is optimal to $R_c'$, provided that $\nabla^2 f(s,b,b)\succ 0$ .

First we claim that $s_1$ cannot be part of an optimizer. 
Indeed, 
$$\lambda_1(s_1)=
-\frac{b \left(5 b c-\sqrt{(b c-9) (b c-1)}+3\right) \left(b c+\sqrt{(b c-9) (b c-1)}-3\right)^2}{4 c^2 \left(b c+3 \sqrt{(b c-9) (b c-1)}-9\right)}
$$
Recall that $bc>9$, and hence the denominator of $\lambda_1(s_1)$ as well as $b c+\sqrt{(b c-9) (b c-1)}-3$ are strictly positive. But then, the sign of $\lambda_1(s_1)$ is exactly the opposite of $5 b c-\sqrt{(b c-9) (b c-1)}+3$. 
Define function $h(x):=5 x-\sqrt{(x-9) (x-1)}+3$ over the domain $x>9$. It is easy to verify that $h(x)$ preserves positive sign (in fact $\min_{x\geq 9} h(x) = h\left( \frac{5}{3} \left(\sqrt{6}+3\right) \right) = 8 \sqrt{6}+28>0$
Hence, $\lambda_1(s_1)<0$ that concludes our claim. 

Next we investigate the spectrum of $c$ for which all $\lambda_i(s_2)$ remain non-negative. 

Our next claim is that for all $bc>9$ we have that $s_2(c)<b/3$. Indeed, consider function 
$$d(x):=3 \sqrt{x^2-10 x+9}-b x+9.
$$
It is easy to see that $d(bc)=6c\left(  \frac{b}3 - s_2(c)  \right)$. 
But then, elementary calculations show that 
$\min_{x\geq 9} d(x) = d(9) =0$, proving that $s_2(c)<b/3$ as claimed. 

Next we investigate the sign of $\lambda_1(s_2), \lambda_2(s_2), \lambda_3(s_2)$. For this, introduce function $t(x) = \sqrt{(x-9) (x-1)}$, and note that 
\begin{align*}
\lambda_1(s_2)&=
-\frac{b \left(-b c+t(bc)+3\right)^2 \left(5 b c+t(bc)+3\right)}{4 c^2 \left(b c-3 \left(t(bc)+3\right)\right)},
\\
\lambda_2(s_2)&=
\frac{30 \left(t(bc)+3\right)+b c \left(b c \left(-3 b c+3 t(bc)+32\right)-5 \left(5 t(bc)+11\right)\right)}{4 c (b c-9)},
\\
\lambda_3(s_2)&=
\frac{b c \left(-23 t(bc)+b c \left(-3 b c+3 t(bc)+22\right)+49\right)-12 \left(t(bc)+3\right)}{4 c (b c-9)}.
\end{align*}

Claim 1: $\lambda_1(s_2) >0$ for all $c>9/b$. \\
Define $d_1(x)=x - 3 (3 + t (x))$ and $d_2(x)=3+5x +t(x)$. Note that $sign\left(\lambda_1(s_2) \right)= - sign( d_1(bc)) \cdot sign(d_2(bc))$. 
Simple calculus shows that $d_1(x)$ is strictly decreasing in $x\geq 9$, and $d_1(9)=0$, and therefore  $d_1(bc)<0$ for all $c>9/b$. Similarly, it is easy to see that $d_2(x)$ is strictly increasing in $x>9$, and $d_2(9)=45$. Therefore $d_2(bc)>0$ for all $c>9/b$. Overall this implies that $\lambda_1(s_2)$ is positive for all $c>9/b$. 

Claim 2: $\lambda_2(s_2) >0$ for all $c\in ( 9/b, 9.72307/b)$. \\
First we observe that the denominator of $\lambda_2(s_2)$ preserves positive sign for $c>9/b$. So we focus on the sign of the numerator we we abbreviate by 
$d_3(x)=30 (3 + t(x)) + x (x (32 - 3 x + 3 t(x)) - 5 (11 + 5 t(x)))$.
Note that $d_3(x) = 0$ is equivalent to that 
\begin{align*}
 & \left( 3 x^2-25 x+30 \right) t(x)-\left( 3 x^3-32 x^2+55 x-90\right) = 0  \\
\Leftrightarrow &
\left( 3 x^2-25 x+30 \right)^2 t^2(x)-\left( 3 x^3-32 x^2+55 x-90\right)^2 = 0\\
\Leftrightarrow &
-8 (x-9) x (3 x (x (2 x-25)+60)-175)=0
\end{align*}
Degree-3 polynomial $3 x (x (2 x-25)+60)-175$ has only one real root, which is 
$$
\frac{1}{18} \left(3 \sqrt[3]{5 \left(192 \sqrt{10}+1055\right)}+\sqrt[3]{142425-25920 \sqrt{10}}+75\right)
\approx 9.72307
$$
Hence, $\lambda_2(s_2) >0$ for all $c\in ( 9/b, 9.72307/b)$

Claim 3: $\lambda_3(s_2) >0$ for all $c\in ( 9/b, 9.06609/b)$. \\
First we observe that the denominator of $\lambda_3(s_2)$ preserves positive sign for $c>9/b$. So we focus on the sign of the numerator we we abbreviate by 
$d_4(x)=x (x (3 t(x)-3 x+22)-23 t(x)+49)-12 (t(x)+3)$.
Note that $d_4(x) = 0$ is equivalent to that 
\begin{align*}
 & \left( 3 x^2-23 x-12 \right) t(x)-\left( 3 x^3-22 x^2-49 x+36\right) = 0  \\
\Leftrightarrow &
\left( 3 x^2-23 x-12 \right)^2 t^2(x)-\left( 3 x^3-22 x^2-49 x+36\right)^2= 0\\
\Leftrightarrow &
-16 (x-9) x (3 x (2 (x-9) x-3)+49)=0
\end{align*}
The roots of degree-3 polynomial $3 x (2 (x-9) x-3)+49$ are 
\begin{align*}
\gamma_1&= 3+\sqrt{38} \cos \left(\frac{1}{3} \tan ^{-1}\left(\frac{127}{151 \sqrt{2}}\right)\right)\approx 9.06609 \\
\gamma'&= 3+\sqrt{\frac{57}{2}} \sin \left(\frac{1}{3} \tan ^{-1}\left(\frac{127}{151 \sqrt{2}}\right)\right)+\sqrt{\frac{19}{2}} \left(-\cos \left(\frac{1}{3} \tan ^{-1}\left(\frac{127}{151 \sqrt{2}}\right)\right)\right) \approx 0.916629 \\
\gamma''&= 3-\sqrt{\frac{57}{2}} \sin \left(\frac{1}{3} \tan ^{-1}\left(\frac{127}{151 \sqrt{2}}\right)\right)+\sqrt{\frac{19}{2}} \left(-\cos \left(\frac{1}{3} \tan ^{-1}\left(\frac{127}{151 \sqrt{2}}\right)\right)\right) \approx -0.982723\\
\end{align*}
We conclude that 
$\lambda_3(s_2)$ preserves positive sign for all $c\in ( 9/b, \gamma_1/b)$.

Overall, we have shown that feasible solution 
$s_0=\frac{- \sqrt{b^2 c^2-10 b c+9}+b c-3}{2 c}, r_0=k_0=b$
satisfies necessary 1st order optimality conditions. We proceed by checking that $s_0, r_0, k_0$ satisfy 2nd order sufficient conditions, which amounts to showing that  
$\nabla^2 f(s_0,b,b)\succ 0$. 
Indeed, 
$$
\nabla^2 f(s_0,b,b)=
\frac{b^3}{(b-s_0)^3}
\left(
\begin{array}{ccc}
 8 & -\frac{\left(b+s_0\right) \left(b^2+4 s_0 b-s_0^2\right)}{b^3} & 4-\frac{4 s_0}{b} \\
 -\frac{\left(b+s_0\right) \left(b^2+4 s_0 b-s_0^2\right)}{b^3} & \frac{2 \left(b^3-s_0 b^2+3 s_0^2 b+s_0^3\right)}{b^3} & \frac{4 s_0 \left(s_0-b\right)}{b^2} \\
 4-\frac{4 s_0}{b} & \frac{4 s_0 \left(s_0-b\right)}{b^2} & \frac{4 \left(b-s_0\right){}^2}{b^2} \\
\end{array}
\right)$$
By setting $q:=s_0/b = \frac{-\sqrt{b^2 c^2-10 b c+9}+b c-3}{2 b c}$, we obtain the simpler form 
\begin{equation}
\nabla^2 f(s_0,b,b)=
\frac{b^3}{(b-s_0)^3}
\left(
\begin{array}{ccc}
 8 & (-q-1) \left(-q^2+4 q+1\right) & 4-4 q \\
 (-q-1) \left(-q^2+4 q+1\right) & 2 \left(q^3+3 q^2-q+1\right) & 4 (q-1) q \\
 4-4 q & 4 (q-1) q & 4 (1-q)^2 \\
\end{array}
\right)
\label{eq: simplified matrix}
\end{equation}
When $bc>9$ we have that $q<1/3$, $q$ is decreasing in the product of $bc>9$, and it remains positive. 
The eigenvalues of the matrix that depends only on $q$ and for any $q\in (0,1/3]$ can be obtained using a closed formula (they are real roots of a degree-3 polynomial). In Figure~\ref{fig: eigen} we depict their behavior. 
Since all eigenvalues are all positive, the candidate optimizer is indeed a minimizer. 
\begin{figure}[h!]
  \centering
  \includegraphics[width=0.5\linewidth]{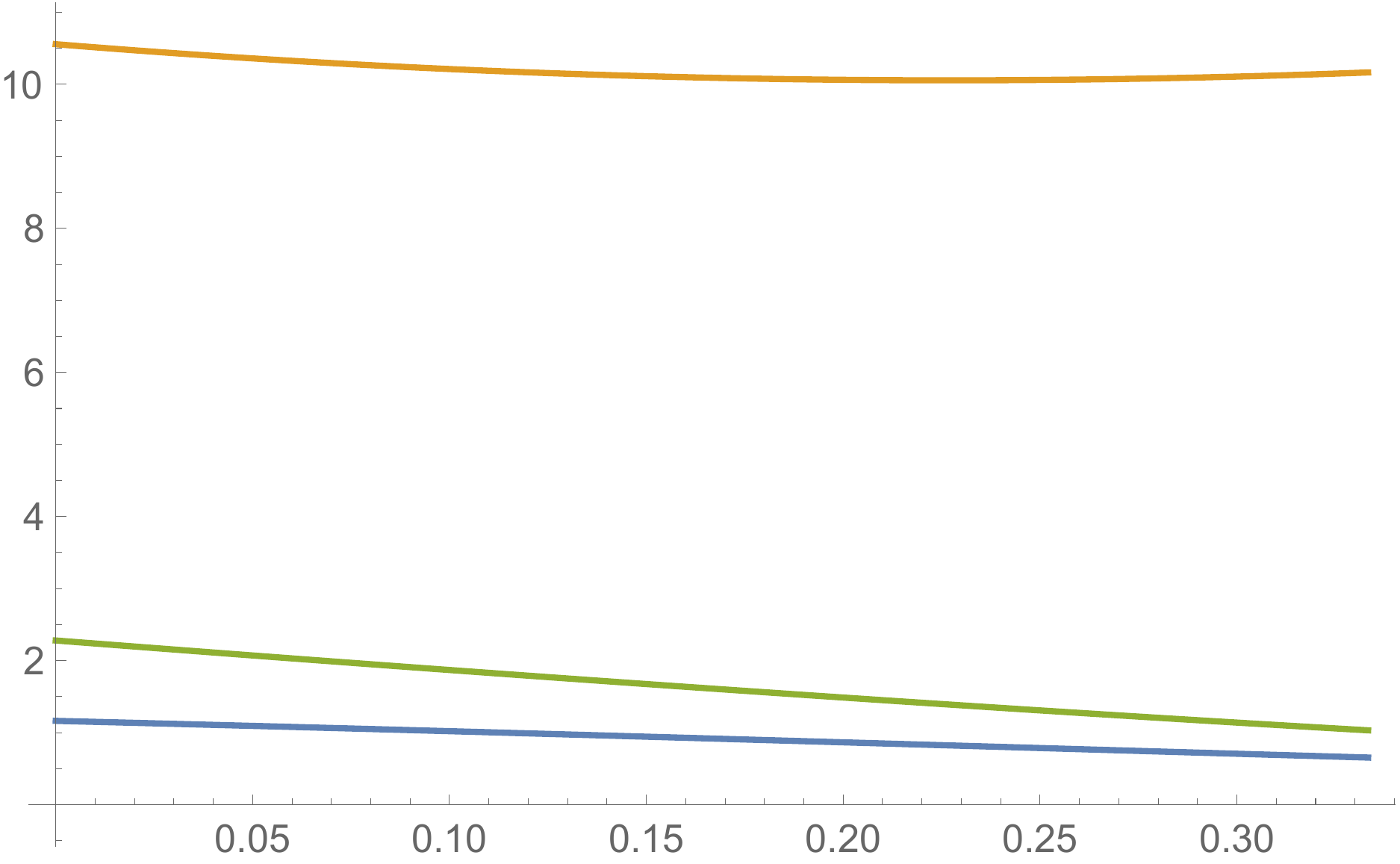}
\caption{
The eigenvalues of matrix~\eqref{eq: simplified matrix} as a function of $q\in (0,1/3]$ (and scaled by $b^3/(b-s_0)^3$). }
\label{fig: eigen}
\end{figure}

\end{proof}

\subsection{Lemma~\ref{lem: optimal speeds naive bounded speeds b>=c>=r1}}
\label{sec: optimal speeds naive bounded speeds b>=c>=r1}
\begin{proof}[Proof of Lemma~\ref{lem: optimal speeds naive bounded speeds b>=c>=r1}]
By Theorem~\ref{thm: optimal speeds naive unbounded speeds}, we know the optimizers to $\nlp{c}{\infty}$; 
$
s=\sigma/c, 
r=\rho/c, 
k=\kappa/c,
$.
These optimizers satisfy the speed bound constraints $s,r,k\leq b$ as long as $\max\{s,r,k\}\leq b$, i.e. $\rho/c\leq b$. Hence, when $c\geq \rho/b$, Non Linear Programs $\nlp{c}{\infty}$, $\nlp{c}{b}$ have the same optimizers. 
\end{proof}

\subsection{Lemma~\ref{lem: naive case between gamma1 gamma2}}
\label{sec: naive case between gamma1 gamma2}
\begin{proof}[Proof of Lemma~\ref{lem: naive case between gamma1 gamma2}]
First, we observe that constraint $r\leq b$ is tight for the provable optimizers for all $c,b$ when $cb \in [9, \gamma_1] \cup \{\gamma_2\}$. As the only other constraint that switches from being tight to non-tight in the same interval is $k\leq b$, we are motivated to maintain tightness for constraints $r\leq b$ and the time constraint. 

Given that speed $s$ is chosen (to be determined later), we fix $r=b$, and we set $k=k_2(c,b)$ where 
$$
k_2(c,b) := \frac{2 b s}{b c s-b-c s^2-s}
$$
so as to satisfy the time constraint tightly (by solving $\ttime{s,r,b}=c$ for $k$). 
It remains to determine values for speed $s$. To this end, we heuristically require that $s=s_2(c,b)$ where 
$$
s_2(c,b):=\alpha \cdot c + \beta
$$
for some constants $\alpha, \beta$ that we allow to depend on $b$. 
In what follows we abbreviate $s_2(c,b)$ by $s_2(c)$. 
Let $s_1(c), s_3$ be the chosen values for speed $s$ as summarized by the statement of Theorem~\ref{thm: (sub)optimal speeds naive bounded speeds} when $cb\leq \gamma_1$ and $cb \geq \gamma_2$, respectively. 
We require that 
$$
s_2(\gamma_1/b) = s_1(\gamma_1/b), ~~s_2(\gamma_2/b) = s_3(\gamma_2/b)
$$
inducing a linear system on $\alpha, \beta$. By solving the linear system, we obtain 
\begin{align*}
\alpha &= 	\frac{b^2 \left(-2 \gamma _1 \sigma +\gamma _2 \gamma _1-\sqrt{\gamma _1^2-10 \gamma _1+9} \gamma _2-3 \gamma _2\right)}{2 \gamma _1 \left(\gamma _1-\gamma _2\right) \gamma _2} \\
\beta &=\frac{b \left(-2 \gamma _1^2 \sigma +\gamma _2^2 \gamma _1-\sqrt{\gamma _1^2-10 \gamma _1+9} \gamma _2^2-3 \gamma _2^2\right)}{2 \gamma _1 \gamma _2 \left(\gamma _2-\gamma _1\right)}
\end{align*}
Using the known values for $\gamma_1, \gamma_2, \sigma$, we obtain $s_2(c)=0.532412 b - 0.0262661 b^2 c$, as promised. 
It remains to argue that $s_2(c,b)$, together with $r=b$, and $k_2(c,b)$ are feasible when $\gamma_1 < cb < \gamma_2$. 

The fact that $s_2(c)$ complies with bounds $0\leq s\leq b$ follows immediately, since $s_2(c)$ is a linear strictly decreasing function  in $c$, and both $s_2(\gamma_1/b), s_2(\gamma_2/b)$ satisfy the bounds by construction. We are therefore left with checking 
that $0 \leq k_2(c,b) \leq b$ which is equivalent to that 
\begin{align*}
& ~~bcs-b-cs^2 - 3s \geq 0 \\
\Leftrightarrow 
& ~~b (b c (b c (0.00170268\, -0.000689908 b c)+0.327748)-1.59724) \geq 0
\end{align*}
Define degree-2 polynomial function $g(x) = x (x (0.00170268\, -0.000689908 x)+0.327748)-1.59724$ and observe that it sufficies to prove that $g(x)\geq 0$ for all $x \in (\gamma_1, \gamma_2)$. The roots of $g$ can be numerically computed as $-22.8094, , 5.0074,  20.2699$, proving that $g$ preserves positive sign in $(\gamma_1, \gamma_2)$ as wanted. 

Finally, the claims regarding the induced energy and competitive ratio is implied by Theorem~\ref{thm: naive as NLP} and obtained by evaluating the given choices of $s,r,k$ in $\eenergy{s,r,k}$. 
\end{proof}

\section{Proofs Omitted from Section~\ref{sec: better algo}.}

\subsection{Proposition~\ref{prop:time-round}}
\label{sec:time-round}

\begin{proof}[Proof of Proposition~\ref{prop:time-round}]
The exploration algorithm explicitly ensures that round $k-1$ ends (and hence that round $k$ begins) at the claimed position, so only the time needs to be proved.
Note that the statement is true for both robots when $k = 0$. 
Suppose that when $L$ starts its $(k-1)$-th round, it is at position $-4^{k-1}$ at time
$3 \cdot 4^{k-1}$.  The round ends at time 
$3 \cdot 4^{k-1} + 4^{k-1} + \frac{1}{s} \cdot 4^{k} \cdot \frac{s}{1-s} + (4^k - 4^{k} \cdot \frac{s}{1-s}) = 
3 \cdot 4^k$, and the $k$-th round will start at the claimed time.  The proof is identical for $R$.
\end{proof}

\subsection{Lemma~\ref{lem:meeting-cases}}
\label{sec:meeting-cases}
\begin{proof}[Proof of Lemma~\ref{lem:meeting-cases}]
\noindent
\emph{Case 1:} $d \in D_1$.
Let $Y = \{L,R\} \setminus \{X\}$ be the robot other than $X$.
Observe that by Proposition~\ref{prop:time-round}, $X$ finds the exit at time $3K + K + d/s = 4K + d/s$.
Then $X$ goes towards $Y$ at speed $s' = \frac{d}{4K - d/s}$, and so it will reach position $0$ at time 
$4K + d/s + \frac{4K - d/s}{d}d = 8K$.
We know that at time $6K$, robot $Y$ is starting a round at position $2K$ or $-2K$, then goes towards $0$ at full speed.  
Hence $Y$ is at position $0$ at time $8K$, where it meets $X$.  

\noindent
\emph{Case 2:} $d \in D_2 = [4Ks/(s + 1), 4Ks/(1-s)]$.
As before, the exit is found at time $4K + d/s$.
Assume for simplicity that $X = L$ (the case $X = R$ is identical by symmetry).  After finding the exit at position $-d$,
$L$ goes full speed to the right.  Thus at time $t = 4K + d/s + d + \frac{d + ds - 4Ks}{1-s}$, 
it arrives at position $\frac{d + ds - 4Ks}{1-s}$.  We show that at this time, $R$ is in its second phase and is at this position.
Notice that 

\begin{align*}
t &= 4K + d/s + d + \frac{d + ds - 4Ks}{1-s} \\
&= 8K + d/s + d - 4K + \frac{d + ds - 4Ks}{1-s} \\
&= 8K + \frac{d/s + d - 4K}{1 - s} \\
&\leq 8K(1 + s/(1-s)^2)
\end{align*}

the latter inequality being obtained from $d \leq 4Ks/(1-s)$.
Now, $R$ enters its second $travel$ phase when at position $0$ at time $8K$, and the phase ends a time 
$8K + 1/s \cdot 8Ks/(1-s) = 8K(1 + 1/(1-s))$.  Since $s \leq 1/2$, we get $8K(1 + 1/(1-s)) \geq 8K(1 + s/(1-s)^2) \geq t$.
Therefore $R$ is still in its second phase at time $t$, and it follows that its position at this time is 
$s(t - 8K) = \frac{d + ds - 4Ks}{1-s}$.  Hence $L$ and $R$ meet.
It is straightforward to see that $L$ and $R$ could not have met before time $t$, and thus $L$ and $R$ meet at the claimed time and position.
\noindent
\emph{Case 3:} $d \in D_3 = [4Ks/(1-s),4K]$.
Again assume that $X = L$.  This time $L$ finds the exit at position $-d$ at time 
$3K + K + \frac{1}{s} \cdot \frac{4Ks}{1-s} + d - \frac{4Ks}{1-s} = 8K + d$.  
Going full speed to the right, at time $t = 8K + 2d + 2ds/(1-s)$, 
it reaches position $2ds/(1-s)$.
Since $d \leq 4K$, we have $t \leq 8K(2 + s/(1-s)) = 8K(1 + 1/(1-s))$. 
As in the previous case, the second phase of $R$ ends at time $8K(1 + 1/(1-s)) \geq t$.
Thus at time $t$, $R$ is at position $s(t - 8K) = 2ds + 2ds^2/(1-s) = 2ds/(1-s)$.
Again, one can check that $L$ and $R$ could not have met before, which concludes the proof.
\end{proof}

\subsection{Lemma~\ref{lem: algo is feasible}}
\label{sec: algo is feasible}
\begin{proof}[Proof of Lemma~\ref{lem: algo is feasible}]
Let $X$ be the robot that finds the exit at distance $d$, and let $K := K(X, k)$.
We show that all speeds are at most 1, as well as that the evacuation time is at most $9d$. There are three cases to consider.

\noindent
\emph{Case 1:} $d \in D_1$.
By Lemma~\ref{lem:meeting-cases}, after meeting, both robots need to travel distance $d$ and have $9d - 8K$ time remaining.
By the last line of the exit phase algorithm, they go back at speed $s_{b1} := \frac{d}{9d - 8K}$ and make it in time, provided that speed $s_{b1}$ is achievable, i.e. $0 < s_{b1} \leq 1$.
Clearly $s_{b1} > 0$ since $9d - 8K > 0$.  If we assume $\frac{d}{9d - 8K} > 1$, we get $d > 9d - 8K$, leading to $d < K$, a contradiction.  

\noindent
\emph{Case 2:} $d \in D_2$.
By Lemma~\ref{lem:meeting-cases}, the robots meet at position $p$ such that $|p| = \frac{d + ds - 4Ks}{1-s}$
at time $t = 8K + \frac{d + d/s - 4K}{1 - s}$.  
The robots use the smallest speed $s_{b2} := \frac{p + d}{t - 9d}$
that allows the them to reach the exit in time $9d$.
We must check that $0 < s_{b2} \leq 1$.  
We argue that if the two robots used speed $1$ to get to the exit after meeting, they would make it before time $9d$.  
Since $s_{b2}$ allows the robots to reach the exit in time exactly $9d$, it follows that $0 < s_{b2} \leq 1$.  

First note that since $d \leq 4Ks/(1-s)$, we have $4K \geq d(1-s)/s$.
Using speed $1$ from the point they meet, the robots would reach the exit at time 

\begin{align*}
8K + \frac{d + d/s - 4K}{1-s} + d + \frac{d + ds - 4Ks}{1-s} &= 4K\left(2 - \frac{1 + s}{1-s}\right) + d \left(1 + \frac{2 + s + 1/s}{1-s}\right) \\
&= 4K \cdot \frac{1-3s}{1-s} + d \cdot \frac{3s + 1}{s(1 - s)} \\
&\leq d \cdot \frac{1-s}{s} \cdot \frac{1-3s}{1-s} + d \cdot \frac{3s + 1}{s(1 - s)}\\
&= d \left( \frac{1-3s}{s} +  \frac{3s + 1}{s(1 - s)} \right)
\end{align*}

where we have used the fact that $1-3s \leq 0$ in the inequality.
It is straighforward to show that $d \left( \frac{1-3s}{1-s} +  \frac{3s + 1}{s(1 - s)} \right) \leq 9d$ 
when $1/3 \leq s \leq 1/2$, proving our claim.

\noindent
\emph{Case 3:} $d \in D_3$.
Again according to Lemma~\ref{lem:meeting-cases}, the robots meet at position $p$ satisfying $|p| = 2ds/(1-s)$ at time 
$t = 8K + 2d + 2ds/(1-s)$. 
The robots go towards the exit at speed $s_{c3} := \frac{p + d}{9d - t}$.
As in the previous case, we show that $s_{c3}$ is a valid speed by arguing that the robots have enough time if they used their full speed.
If they do use speed $1$ after they meet, they reach the exit at time 
$t' = 8K + 3d + 4ds/(1-s)$.  Since $d \geq 4Ks/(1-s)$, we have $K \leq d(1-s)/(4s)$.
Therefore $t' \leq 8 \cdot (d(1-s)/(4s)) + 3d + 4ds/(1-s) = d \cdot \frac{3s^2 - s + 2}{s(1-s)}$.  
One can check that this is $9d$ or less whenever $1/3 \leq s \leq 1/2$. 
\end{proof}

\subsection{Lemma~\ref{lem:energies}}
\label{sec:energies}
\begin{proof}[Proof of Lemma~\ref{lem:energies}]
For any $K' := 4^i$ power of $4$ with $i \geq 1$, define $B_L(K',s)$ as the energy spent by $L$ after reaching position $-K'$ for the first time without having found the exit, ignoring the initial $1/9$ energy spent to get to position $1$.
The quantity $B_L(K',s)$ is the sum of energy spent in each of the first $i - 1$ rounds, and so

\begin{align*}
B_L(K',s) &= \sum_{j = 0}^{i-1} \left(4^j + s^2(4^{j+1}s/(1-s)) + 4^{j+1} - 4^{j+1}s/(1-s)    \right) \\
&= \sum_{j = 0}^{i-1} \left(4^{j+1}(5/4+(s^3 - s)/(1-s)) \right) \\
&= 4/3 (4^i - 1)(5/4 - s(s+1)) \\
&= 1/3 (K' - 1)(5 - 4s(s+1))
\end{align*}

We define $B_R(2K',s)$ similarly for $R$, i.e. $B_R(2K',s)$ is the energy spent by $R$ when its $(i-1)$-th round is finished and it reached position 
$2K'$ for the first time, ignoring the initial $2/9$ energy to get at position $2$.  We get 

\begin{align*}
B_R(2K',s) &= \sum_{j = 0}^{i-1} \left(2 \cdot 4^j + s^2(2 \cdot 4^{j+1}s/(1-s)) + 2 \cdot 4^{j+1} - 2 \cdot 4^{j+1}s/(1-s)    \right) \\
&= 2B_L(K',s) = 2/3 (K' - 1)(5 - 4s(s+1))
\end{align*}

We may now calculate the three possible cases of energy.  Assume that $X \in \{L,R\}$ finds the exit
and $Y = \{L,R\} \setminus \{X\}$.
Observe that 
$$
B_X(K,s) + B_Y(2K,s) = (K - 1)(5 - 4s(s+1)) = F(K,s)
$$

We implictly use Lemma~\ref{lem:meeting-cases} for the distance traveled by $X$ to catch up to $Y$ after
finding the exit, and the distance traveled back by both robots.
In the $E(K,d,s)$ expressions that follow, for clarity we partition the terms into 3 brackets, which respectively represent the energy spent by $X$ to find the exit and catch up to $Y$, the energy spent by $Y$ before being caught, 
and the energy spent by both robots to go to the exit.

\noindent
\emph{Case 1: $d \in D_1$}.
The total energy spent is 
\begin{align*}
E(K,d,s) &= \left[ B_X(K,s) + K + s^2d + s_{c1}^2 d \right] + \left[ B_Y(2K,s) + 2K \right] + \left[ 2 s_{b1}^2 d \right] \\
\\&= F(K,s) + 3K + d(s^2 + s_{c1}^2 + 2s_{b1}^2)
\end{align*}

\noindent
\emph{Case 2: $d \in D_2$}.
In this case, the energy spent is 

\begin{align*}
E(K,d,s) &= \left[ B_X(K,s) + K + s^2 d + d + \frac{d + ds - 4Ks}{1-s} \right]  + \left[ B_Y(2K,s) + 2K + s^2 \left( \frac{d + ds - 4Ks}{1-s} \right)  \right] +  \\
 & \quad  \left[2 s_{b2}^2 \left( d + \frac{d + ds - 4Ks}{1-s} \right)\right] \\
&= F(K,s) + 3K + \left( \frac{2d - 4Ks}{1-s}\right)(1 + s^2 + 2s_{b2}^2)
\end{align*}

\noindent
\emph{Case 3: $d \in D_3$}.
The energy spent is

\begin{align*}
E(K,d,s) &= \left[B_X(K,s) + K + s^2 4Ks/(1-s) + d - 4Ks/(1-s) + d +  2ds/(1-s) \right] + \\
& \quad \left[B_Y(2K,s) + 2K + s^2(2ds/(1-s)) \right] + \left[2s^2_{b3}\left(d + 2ds/(1-s)  \right)    \right] \\
&= F(K,s) + 3K  + 4Ks/(1-s) (s^2 - 1) + d \left( \frac{s}{1-s}(2 + 2s^2 + 2s_{b3}^2) + 2 + 2s_{b3}^2 \right) \\
&= F(K,s) + 3K - 4Ks(s + 1) + \frac{2d}{1-s}(s^3 + s_{b3}^2(s + 1) + 1)
\end{align*}
\end{proof}

\subsection{Lemma~\ref{lem: monotonicity of algo s}}
\label{sec: monotonicity of algo s}
\begin{proof}[Proof of Lemma~\ref{lem: monotonicity of algo s}]
We analyze each case separately.

\noindent
\emph{Case 1: } $d \in D_1 = [K, 4Ks/(s + 1)]$.  
According to Lemma~\ref{lem:energies}, we have

\begin{align*}
E(K,d,s)/d &= 1/d \cdot ((K - 1)(5 - 4s(s+1)) + 3K + d(s^2 + s_{c1}^2 + 2s_{b1}^2)) \\
&= \dfrac{(K - 1)(5 - 4s(s+1)) + 3K}{d} + s^2 + \left( \frac{d}{4K - d/s} \right)^2 + 2\left( \frac{d}{9d - 8K} \right)^2
\end{align*}

Consider the case $d = K$.  Plugging in $s = 0.39403$, the above evaluates to $8.4258714091 - 2.8028414364/K$.  
We claim that $E(K,d,s)/d$ is a decreasing function over interval $D_1$, and therefore attains its maximum when $d=K$.  Assuming this is true, adding the initialization energy of $1/3$ omitted so far and given that $d \geq K$, the energy ratio is at most  

$$
8.4258714091 - 2.8028414364/K + 1/(3K) \leq 8.42588
$$


We now prove that $E(K,d,s)/d$ is decreasing over the interval $D_1$.
Let 
\begin{align*}
f_1(K,d,s):=&\frac{d^2 }{(9 d-8 K)^2} \\
f_2(K,d,s):=&\frac{d^2 }{(d-4 K s)^2}\\
f_3(K,d,s):=&\frac{-4 K \left(s^2+s-2\right)+4 s (s+1)-5}{d},
\end{align*}
and observe that 
$$
\frac{E(K,d,s)}{d} =2 f_1(K,d,s)+s^2 f_2(K,d,s)+f_3(K,d,s)+s^2.
$$

The plan is to prove that 
$$
\frac{\partial}{\partial d} E(K,d,s)/d <0.
$$

For this we calculate 
\begin{align*}
\frac{\partial}{\partial d} f_1(K,d,s):= 
& -\frac{16 d K}{(9 d-8 K)^3} \\
\frac{\partial}{\partial d} f_2(K,d,s):=& 
-\frac{8 d K s}{(d-4 K s)^3}\\
\frac{\partial}{\partial d} f_3(K,d,s):=&
\frac{4 K \left(s^2+s-2\right)-4 s (s+1)+5}{d^2},
\end{align*}

Now we claim that all $\frac{\partial}{\partial d} f_i(K,d,s)$ are increasing functions in $d$, for $i=1,2,3$. 
Indeed, first, 
$$
\frac{\partial^2}{\partial d^2} f_1(K,d,s)
=
\frac{32 K (9 d+4 K)}{(9 d-8 K)^4}>0
$$
since $d\geq K$. Hence $\frac{\partial}{\partial d} f_1(K,d,s)$ is increasing in $d$.

Second, 
$$
\frac{\partial^2}{\partial d^2} f_2(K,d,s)
=
\frac{16 K s (d+2 K s)}{(d-4 K s)^4}
$$
is positive (and well defined), since $d\leq 4ks/(1+s)$. Hence $\frac{\partial}{\partial d} f_2(K,d,s)$ is increasing in $d$.

Third, we show that $\frac{\partial}{\partial d} f_3(K,d,s)$ is increasing in $d$. For this it is enough to prove that 
$
4 K \left(s^2+s-2\right)-4 s (s+1)+5 <0.
$
For $s=0.39403$ (and in fact for all $s\in (-2,1)$) the strict inequality can be written as 
$
K\geq \frac{4 s^2+4 s-5}{4 s^2+4 s-8},
$
which we show next it is satisfied. 
Indeed, it is easy to see that 
$\frac{4 s^2+4 s-5}{4 s^2+4 s-8} \leq 5/8$ (which is attained for $s=0$), while $K\geq 4$, hence the claim follows. 

To resume, we showed that $\frac{\partial}{\partial d} f_i(K,d,s)$ are increasing functions in $d$, for $i=1,2,3$. Recalling that $s=0.39403$, and since $d\leq 4ks/(1+s)$, we obtain that 
\begin{align*}
&\frac{\partial}{\partial d} E(K,d,s)/d \\
& \leq 
2 
\frac{\partial}{\partial d} f_1(K,4Ks/(1+s),s)
+
s^2 
\frac{\partial}{\partial d}
f_2(K,4Ks/(1+s),s)
+
\frac{\partial}{\partial d}
f_3(K,4Ks/(1+s),s) \\
&= \frac{(s+1)^2 (s (4 K (s (s+1) (49 s (7 s-6)+76)-8)-s (7 s (28 s (7 s+1)-365)+1774)+452)-40)}{16 K^2 s^2 (7 s-2)^3}\\
&= \frac{2.19262\, -1.79465 K}{K^2}.
\end{align*}
Since $K\geq 4$, the latter quantity is clearly negative. This shows that 
$\frac{\partial}{\partial d} E(K,d,s)/d$ is negative (in the given domain), hence $E(K,d,s)/d$ is decreasing in $d$.


\vspace{4mm}

\noindent
\emph{Case 2: } $d \in D_2 = [4Ks/(s + 1), 4Ks/(1-s)]$.
In this case,  the energy ratio $E(K,d,s)/d$ is

\begin{align*}
& \quad 1/d \cdot \left( 1/3 + (K - 1)(5 - 4s(s+1)) + 3K + \left( \frac{2d - 4Ks}{1-s}\right)(1 + s^2 + 2s_{b2}^2) \right)\\
&= \frac{(K - 1)(5 - 4s(s+1)) + 3K}{d} +  \left( \frac{2d - 4Ks}{d(1-s)} \right) \left(1 + s^2 + 2\left( \frac{2d - 4Ks}{d(8 - 9s - 1/s) + 4K(2s-1)} \right)^2 \right)
\end{align*}

We will show that this expression achieves its maximum at $d = 4Ks/(1-s)$.
When $s = 0.39403$, then above yields $8.425786060 - 1.0776069241/K$.  Given that $1/(3d) \leq 1/(3.39186 K)$, this implies that the energy ratio is at most 

$$
8.425786060 - 1.0776069241/K + 1/(3.39186 K) \leq 8.42588
$$

We prove that $E(K,d,s)/d$ is an increasing function over interval $D_2$.
First we compute 
$
\frac{\partial }{\partial d} \frac{E(K,d,s)}{d}
$
and we substitute $s=0.39403$ to find $
\frac{1}{d^2 (d-0.442471 K)^3} g(K,d)
$, where

\begin{align*}
g(K,d):= &d^3 (7.8438 K+2.80279)+d^2 K (-15.5674 K-3.72046) \\
&+d K^2 (8.91957 K+1.6462)+K^3 (-1.31555 K-0.242798).
\end{align*}
Note that $d\geq 4Ks/(1+s) \approx 1.13064 K$, and hence $d^2 (d-0.442471 K)^3 >0$ for all values of $d$ under consideration. Therefore the lemma will follow if we show that $g(K,d)\geq 0$ as well. 

$g(K,d)$ is a degree-3 polynomial with positive leading coefficient. It attains a local minimum at the largest real root of 
$$\frac{\partial }{\partial d} g(K,d)
=3 d^2 (7.8438 K+2.80279)+2 d K (-15.5674 K-3.72046)+K^2 (8.91957 K+1.6462)
$$
which is
$$
d_0(K) := 
\frac{K \left(0.661559 K + 0.0212482 \sqrt{K (129.818 K+8.39821)}+0.158106\right)}{ K+0.357325}
$$
Now we observe that for all $K>0$, we have 
$d_0(K) < 4Ks/(1+s) \approx 1.13064 K$.

From the above, it follows that $g(K,d)$ is monotonically increasing for $d\geq 4Ks/(1+s)$, and therefore 
$$
g(K,d) \geq 
g(K, 4Ks/(1+s))
= (0.205742 K+0.913416) K^3
\geq 0
$$
as wanted.


\vspace{4mm}

\noindent
\emph{Case 3: } $d \in D_3 = [4Ks/(1-s),4K]$.
The energy ratio $E(K,d,s)/d$ is 

\begin{align*}
&\quad  1/d \cdot \left( 1/3 + (K - 1)(5 - 4s(s+1)) + 3K - 4Ks(s + 1) + \frac{2d}{1-s}(s^3 + s_{b3}^2(s + 1) + 1) \right) \\
&= \frac{1/3 + (K - 1)(5 - 4s(s+1)) + 3K - 4Ks(s+1)}{d} +  \frac{2s^3 + 2}{1-s} + \frac{2s+2}{1-s} \left( \frac{d(1 + s)}{d(7 - 9s) + 8K(s - 1)} \right)^2 
\end{align*}

In this case, we claim that this expression is decreasing over $D_3$ and achieves its maximum at $d = 4Ks/(1-s)$.
When $s = 0.39403$, the above gives $8.425786060 - 1.0776069241/K$ (which is the same as in case 2, as one should expect).  Given that $1/(3d) \leq 1/(7.80296 K)$, we get that the energy ratio is at most 

$$
8.425786060 - 1.0776069241/K + 1/(7.80296K) \leq 8.42588
$$

Let us prove that $E(K,d,s)/d$ is indeed decreasing.
Note that $E(K,d,s)/d$ equals

$$
\frac{2 (s+1)^3}{1-s}
g(K,d,s)
+ 
h(K,d,s)
+\frac{2 \left(s^3+1\right)}{1-s}
$$

where 
$$
g(K,d,s):=\frac{d^2}{(d (7-9 s)+8 K (s-1))^2}, 
~~
h(K,d,s):=\frac{8 K \left(-s^2-s+1\right)+4 s (s+1)-5}
{d}.
$$
In what follows we prove that both $g(K,d,s), h(K,d,s) $ are strictly decreasing when $d\geq 4Ks/(1-s)$, implying the claim of the lemma. 

First we show that $h(K,d,s)$ is decreasing. For that note that, using the fixed value of $s = 0.39403$, we have
$
h(K,d,s) = (3.60558 K-2.80279)/d
$, and the latter expression (in $d$) is clearly strictly decreasing for all constants $K>1$. 

Now we show that $g(K,d,s)$ is strictly decreasing for all $d\geq 4Ks/(1-s)$. 
For that observe that for the specific constant $s$, and since $d\geq 4Ks/(1-s)\approx 2.60107 K$, we have that 
$$
|
d (7-9 s)+8 K (s-1)
|
=d (7-9 s)+8 K (s-1) >0.
$$
Hence, to show that $g(K,d,s)$ is strictly decreasing, it is enough to prove that 
$$
q(K,d,s):=\frac{d}{(d (7-9 s)+8 K (s-1))}
$$
is strictly decreasing in $d\geq 4Ks/(1-s)$. First observe that the rational function is well defined for these values of $d$, since the denominator becomes 0 only when $d= \frac{8 K (s-1)}{9 s-7} < 4Ks/(1-s)$ (the last inequality is easy to verify). To that end, we compute 
$$
\frac{\partial }{\partial d} q(K,d,s)
= 
\frac{8 K (s-1)}{(d (7-9 s)+8 K (s-1))^2}
$$
which is of course negative for the given value of $s<1$. 

\end{proof}

\end{document}